\documentclass[MSNbibl,number,citesort,seceqn,dvips]{arxbj}
\usepackage{graphicx}
\aid{0}
\volume{19}
\issue{5A}
\pubyear{2013}
\firstpage{2033}
\lastpage{2066}
\doi{10.3150/12-BEJ442} 

\makeatletter
\def\eqref#1{(\ref{#1})}
\newproclaim{assumption}{Assumption}[section]
\newtheorem{theorem}[assumption]{Theorem}
\newtheorem{prop}[assumption]{Proposition}
\newtheorem{lem}[assumption]{Lemma}
\newtheorem{cor}[assumption]{Corollary}
\newremark{rem}[assumption]{Remark}
\makeatother

\begin{document}
\begin{frontmatter}

\title{Nonasymptotic bounds on the estimation error
of MCMC algorithms}
\runtitle{Nonasymptotic estimation error of MCMC}

\begin{aug}
\author[1]{\fnms{Krzysztof} \snm{\L atuszy\'{n}ski}\corref{}\thanksref{1}\ead[label=e1]{latuch@gmail.com}\ead[label=u1,url]{http://www.warwick.ac.uk/go/klatuszynski}},
\author[2]{\fnms{B{\l}a\.zej} \snm{Miasojedow}\thanksref{2}\ead[label=e2]{bmia@mimuw.edu.pl}}
\and\\
\author[3]{\fnms{Wojciech} \snm{Niemiro}\thanksref{3}\ead[label=e3]{wniem@mat.uni.torun.pl}}
\runauthor{K. \L atuszy\'{n}ski, B. Miasojedow and W. Niemiro} 
\address[1]{Department of Statistics, University of Warwick, CV4 7AL,
Coventry, UK.\\ \printead{e1}; \printead{u1}}
\address[2]{Institute of Applied Mathematics and Mechanics, University
of Warsaw, Banacha 2, 02-097 Warszawa, Poland. \printead{e2}}
\address[3]{Faculty of Mathematics and Computer Science, Nicolaus
Copernicus University, Chopina 12/18, 87-100 Toru\'{n}, Poland.
\printead{e3}}
\end{aug}

\received{\smonth{8} \syear{2011}}
\revised{\smonth{1} \syear{2012}}

%
\begin{abstract}
We address the problem of upper bounding the mean square error of MCMC
estimators. Our analysis is nonasymptotic.
We first establish a general result valid for essentially all ergodic
Markov chains encountered in Bayesian
computation and a possibly unbounded target function~$f$. The bound is
sharp in the sense that the leading term is exactly $\sigma_{\mathrm
{as}}^{2}(P,f)/n$,
where $\sigma_{\mathrm{as}}^{2}(P,f)$ is the CLT asymptotic variance.
Next, we proceed to
specific additional assumptions and give explicit computable bounds for
geometrically and polynomially ergodic Markov chains under quantitative
drift conditions. As a corollary, we provide results on confidence estimation.
\end{abstract}

%

\end{frontmatter}
%

\section{Introduction}\label{sec1}

Let $\pi$ be a probability distribution on a Polish space $\mathcal
{X}$ and
$f\dvtx\mathcal{X}\to\mathbb{R}$ be
a Borel function. The objective is to compute (estimate) the quantity
\[
\theta:=\pi(f)=\int_\mathcal{X}\pi(
\mathrm{d}x)f(x).
\]
%
Typically $\mathcal{X}$ is a high dimensional space, $f$ need not be bounded
and the density of $\pi$ is known up to a normalizing constant. Such
problems arise in Bayesian inference and are often solved
using Markov chain Monte Carlo (MCMC) methods. The idea is to simulate
a Markov chain $(X_{n})$ with transition kernel $P$ such that $\pi
P=\pi$,
that is $\pi$ is stationary with respect to $P$. Then averages
along the trajectory of the chain,
\[
\label{eq: est} \hat\theta_{n}:= \frac{1}{n}\sum
_{i=0}^{n-1}f(X_{i})
\]
are used to estimate $\theta$. It is essential to have explicit and
reliable bounds which provide information about how long
the algorithms must be run to achieve a prescribed level of accuracy
(cf. \cite{rosenthal1995minorization,jones2001honest,jones2006fixed}).
The aim of our paper is to derive nonasymptotic and explicit
bounds on the mean square error,
%
\begin{equation}
\label{eq:rmse} \mathrm{MSE}:= \mathbb{E}(\hat\theta_{n}-\theta
)^{2}.
\end{equation}

To upper bound (\ref{eq:rmse}), we begin with a general inequality
valid for all ergodic Markov chains that admit a one step small set
condition. Our bound is sharp in the sense that the leading term is
exactly $\sigma_{\mathrm{as}}^{2}(P,f)/n$,
where $\sigma_{\mathrm{as}}^{2}(P,f)$ is the asymptotic variance in
the central limit
theorem. The proof relies on the regeneration technique, methods of
renewal theory and statistical sequential analysis.

To obtain explicit bounds, we subsequently consider geometrically and
polynomially ergodic Markov chains. We assume appropriate drift
conditions that give quantitative information about the transition
kernel $P.$ The upper bounds on MSE are then stated in terms of the
drift parameters.

We note that most MCMC algorithms implemented in Bayesian inference are
geometrically or polynomially ergodic (however establishing the
quantitative drift conditions we utilize may be prohibitively difficult
for complicated models). Uniform ergodicity is stronger then
geometrical ergodicity considered here and is often discussed in
literature. However, few MCMC algorithms used in practice are uniformly
ergodic. MSE and confidence estimation for uniformly ergodic chains are
discussed in our accompanying paper~\cite{latuszynski2011nonasymptotic}.

The Subgeometric condition, considered in, for example, \cite
{douc2008bounds}, is more general than polynomial ergodicity considered
here. We note that with some additional effort, the results for
polynomially ergodic chains (Section \ref{Sec:Poly}) can be
reformulated for subgeometric Markov chains. Motivated by applications,
we avoid these technical difficulties.

Upper bounding the mean square error (\ref{eq:rmse}) leads immediately
to confidence estimation by
applying the Chebyshev inequality. One can also apply the more
sophisticated median trick of \cite{jerrum1986random}, further
developed in \cite{niemiro2009fixed}
. The median trick leads to an exponential inequality for the MCMC
estimate whenever the MSE can be upper bounded, in particular in the
setting of geometrically and polynomially ergodic chains. 

We illustrate our results
with benchmark examples. The first, which is related to a simplified
hierarchical Bayesian model and similar to \cite{jones2001honest}, Example 2, allows to compare the bounds provided in our paper
with actual MCMC errors. Next, we demonstrate how to apply our results
in the Poisson--Gamma model of \cite{gelfand1990sampling}. 
Finally, the contracting normals toy-example allows for a numerical
comparison with our earlier work \cite{eps_alpha}.

The paper is organised as follows: in Section~\ref{Sec:Regen} we give
background on the regeneration technique and introduce notation. The
general MSE upper bound is derived in Section~\ref{Sec:MainTh}.
Geometrically and polynomially ergodic Markov chains are considered in
Sections~\ref{Sec:Geom} and~\ref{Sec:Poly}, respectively. The
applicability of our results is discussed in Section~\ref{Sec:Ex},
where also numerical examples are presented. 
Technical proofs are deferred to Sections~\ref{Sec:Lemmas} and~\ref
{Sec:Proofs}.


\subsection{Related nonasymptotic results}

A vast literature on nonasymptotic analysis of Markov chains is
available in various settings. To place our results in this context, we
give a brief account.

In the case of \textit{finite state space}, an approach based on the
spectral decomposition was used in \cite{aldous1987markov,gillman1998chernoff,leon2004optimal,niemiro2009fixed} to derive
results of related type.

For \textit{bounded} functionals of \textit{uniformly} ergodic chains on
a general state space, exponential inequalities with explicit constants
such as those in \cite{glynn2002hoeffding,kontoyiannis2005relative}
can be applied to derive confidence bounds. In the accompanying paper
\cite{latuszynski2011nonasymptotic}, we compare the simulation cost of
confidence estimation based on our approach (MSE bounds with the median
trick) to exponential inequalities and conclude that while exponential
inequalities have sharper constants, our approach gives in this setting
the optimal dependence on the regeneration rate $\beta$ and therefore
will turn out more efficient in many practical examples.

Related results come also from studying concentration of measure
phenomenon for dependent random variables. For the large body of work
in this area see, for example, \cite{marton1996measure,samson2000concentration} and \cite{kontorovich2008concentration} (and
references therein), where transportation inequalities or martingale
approach have been used. These results, motivated in a more general
setting, are valid for Lipschitz functions with respect to the Hamming
metric. They also include expressions $\sup_{x,y \in\mathcal{X}} \|
P^i(x, \cdot) - P^i(y, \cdot)\|_{\mathrm{tv}}$ and when applied to our
setting, they are well suited for \textit{bounded} functionals of \textit
{uniformly} ergodic Markov chains, but cannot be applied to
geometrically ergodic chains. For details, we refer to the original
papers and the discussion in Section 3.5 of \cite{Adamczak2008}.

For lazy reversible Markov chains, nonasymptotic mean square error
bounds have been obtained for \textit{bounded} target functions in \cite
{rudolf2009explicit} in a setting where explicit bounds on conductance
are available. These results have been applied to approximating
integrals over balls in $\mathbb{R}^d$ under some regularity
conditions for
the stationary measure, see \cite{rudolf2009explicit} for details. The
Markov chains considered there are in fact uniformly ergodic, however
in their setting the regeneration rate $\beta,$ can be verified for
$P^h,$ $h > 1$ rather then for $P$ and turns out to be exponentially
small in dimension. Hence, conductance seems to be the natural approach
to make the problem tractable in high dimensions.

Tail inequalities for \textit{bounded} functionals of Markov chains that
are not uniformly ergodic were considered in \cite
{clemenccon2001moment,Adamczak2008} and \cite{douc2008bounds}
using regeneration techniques. These results apply for example, to
\textit{geometrically} or \textit{subgeometrically} ergodic Markov
chains, however they also involve nonexplicit constants or require
tractability of moment conditions of random tours between
regenerations. Computing explicit bounds from these results may be
possible with additional work, but we do not pursue it here.


Nonasymptotic analysis of \textit{unbounded} functionals of
Markov chains is scarce. In particular, tail inequalities for
\textit{unbounded} target function $f$ that can be applied to
\textit{geometrically} ergodic Markov chains have been established by
Bertail and Cl\'emen\c con in \cite{bertail2009sharp} by regenerative
approach and using truncation arguments. However, they involve
nonexplicit constants and can not be directly applied to confidence
estimation. Nonasymptotic and explicit MSE bounds for geometrically
ergodic MCMC samplers have been obtained in \cite{eps_alpha} under a
geometric drift condition by exploiting computable convergence rates.
Our present paper improves these results in a fundamental way.
Firstly, the generic Theorem~\ref{th:BasicMSE} allows to extend the
approach to different classes of Markov chains, for example,
polynomially ergodic in Section~\ref{Sec:Poly}. Secondly, rather then
resting on computable convergence rates, the present approach relies
on upper-bounding the CLT asymptotic variance which, somewhat
surprisingly, appears to be more accurate and consequently the MSE
bound is much sharper, as demonstrated by numerical examples in
Section \ref{Sec:Ex}.

Recent work \cite{joulin2010curvature} address error estimates for
MCMC algorithms under positive curvature condition. The positive
curvature implies geometric ergodicity in the Wasserstein distance and
bivariate drift conditions (cf. \cite{roberts2001small}). Their
approach appears to be applicable in different settings to ours and
also rests on different notions, for example, employs the coarse
diffusion constant instead of the exact asymptotic variance. Moreover,
the target function $f$ is assumed to be Lipschitz which is problematic
in Bayesian inference. Therefore, our results and \cite
{joulin2010curvature} appear to be complementary.

Nonasymptotic rates of convergence of geometrically, polynomially and
subgeometrically ergodic Markov chains to their stationary
distributions have been investigated in many papers \cite
{Meyn1994computable,rosenthal1995rates,roberts1999bounds,rosenthal2002quantitative,jones2004sufficient,fort2003computable,douc2004quantitative,baxendale2005renewal,fort2003polynomial,douc2007computable,robertsquantitative} under assumptions similar to
our Section \ref{Sec:Geom} and \ref{Sec:Poly}, together with an
aperiodicity condition that is not needed for our purposes. Such
results, although of utmost theoretical importance, do not directly
translate into bounds on accuracy of estimation, as they allow to
control only the bias of estimates and the so-called burn-in time.
\section{Regeneration construction and notation}\label{Sec:Regen}

Assume $P$ has invariant distribution $\pi$ on $\mathcal{X}$, is $\pi
$-irreducible and Harris recurrent. The following one step small set
Assumption \ref{as:SmallSetCond} is verifiable for virtually all
Markov chains targeting Bayesian posterior distributions. It allows for
the regeneration/split construction of Nummelin \cite
{nummelin1978splitting} and Athreya
and Ney \cite{athreya1978new}.

\begin{assumption}[(Small set)]\label{as:SmallSetCond}
There exist a Borel set $J\subseteq\mathcal{X}$ of positive $\pi$
measure, a
number $\beta>0$ and
a probability measure $\nu$ such that
\[
P(x,\cdot)\geq\beta\mathbb{I}(x\in J)\nu(\cdot).
\]
\end{assumption}

Under Assumption \ref{as:SmallSetCond}, we can define a bivariate
{Markov chain} $(X_n,\Gamma_n)$
on the space $\mathcal{X}\times\{0,1\}$ in the following way. Bell variable
$\Gamma_{n-1}$ depends only on $X_{n-1}$ via
%
\begin{equation}
\label{eq:gamma} \mathbb{P}(\Gamma_{n-1}=1|X_{n-1}=x)=\beta
\mathbb{I}(x\in J).
\end{equation}
The rule of transition from $(X_{n-1},\Gamma_{n-1})$ to $X_{n}$ is
given by
\begin{eqnarray*}
\mathbb{P}(X_{n}\in A|
\Gamma_{n-1}=1,X_{n-1}=x)&=&\nu(A),
\\
\mathbb{P}(X_{n}\in A|\Gamma_{n-1}=0,X_{n-1}=x)&=&Q(x,A),
\end{eqnarray*}
where $Q$ is the normalized ``residual'' kernel given by
\[
Q(x,\cdot):=\frac{P(x,\cdot)- \beta\mathbb{I}(x\in J)\nu(\cdot
)}{1-\beta\mathbb{I}(x\in J)}.
\]
Whenever $\Gamma_{n-1}=1$, the chain regenerates at moment $n$.
The regeneration epochs are
\begin{eqnarray*}
  T&:=&T_1:=\min\{n\geq1\dvt
\Gamma_{n-1}=1\},
\\
T_k&:=&\min\{n\geq T_{k-1}: \Gamma_{n-1}=1\}.
\end{eqnarray*}
Write $\tau_{k}:=T_{k}-T_{k-1}$ for $k=2,3,\ldots$ and $\tau_1:=T$.
Random blocks
\begin{eqnarray*}
  \Xi&:=&\Xi_1 :=(X_{0},
\ldots,X_{T-1},T),
\\
\Xi_k &:=&(X_{T_{k-1}},\ldots,X_{T_{k}-1},
\tau_{k})
\end{eqnarray*}
for $k=1,2,3,\ldots$ are independent.

We note that numbering of the bell variables $\Gamma_n$ may differ
between authors: in our notation $\Gamma_{n-1} = 1$ indicates
regeneration at moment $n$, not $n-1$. Let symbols $\mathbb{P}_{\xi}$ and
$\mathbb{E}_{\xi}$ mean that $X_{0}\sim\xi$. Note also that these symbols
are unambiguous, because specifying the distribution of $X_0$ is
equivalent to specifying the joint distribution of $(X_0, \Gamma_0)$
via~(\ref{eq:gamma}).

For $k=2,3,\ldots,$ every block
$\Xi_{k}$ under $\mathbb{P}_{\xi}$ has the same distribution as $\Xi$
under $\mathbb{P}_{\nu}$. However, the distribution of
$\Xi$ under $\mathbb{P}_{\xi}$ is in general different. We will also use
the following notations
for the block sums:
\[
  \Xi(f) :=\sum_{i=0}^{T-1}
f(X_i), \qquad \Xi_{k}(f) :=\sum
_{i=T_{k-1}}^{T_{k}-1} f(X_i).
\]

\section{A general inequality for the MSE}\label{Sec:MainTh}

We assume that $X_0\sim\xi$ and thus $X_n\sim\xi P^n$. Write ${\bar f}
:=f-\pi(f)$.

\begin{theorem}\label{th:BasicMSE}
If Assumption \ref{as:SmallSetCond} holds, then
%
\begin{equation}
\label{eq:BasicMSE_bound} \sqrt{\mathbb{E}_{\xi} (\hat
\theta_{n}-\theta)^2}\leq\frac{\sigma_{\mathrm{as}}(P,f)}{\sqrt
{n}} \biggl( 1+
\frac{C_{0}(P)}{n} \biggr) +\frac{C_{1}(P,f)}{n}+\frac{C_{2}(P,f)}{n},
\end{equation}
where
%
\begin{eqnarray}
\label{eqdef:asvar} \sigma_{\mathrm{as}}^{2}(P,f)& := &
\frac{\mathbb{E}_{\nu} (\Xi({\bar
f}))^{2}}{\mathbb{E}_{\nu}T} ,
\\
\label{eqdef:czero} C_{0}(P)&:= & \mathbb{E}_{\pi} T-
\frac{1}{2},
\\
\label{eqdef:cone} C_{1}(P,f)& := & \sqrt{\mathbb{E}_{\xi} \bigl(
\Xi\bigl(|{\bar f}|\bigr)\bigr)^{2}},
\\
\label{eqdef:ctwo} C_{2}(P,f)&=&C_2(P,f,n) :=  \sqrt
{\mathbb{E}_{\xi} \Biggl(\mathbb{I}(T_1 < n)\sum
_{i=n}^{T_{R(n)}-1}|{\bar f}|(X_i)
\Biggr)^{2}},
\\
\label{eqdef:Rn} R(n) & := & \min\{r\geq1\dvt T_{r}> n\}.
\end{eqnarray}
\end{theorem}

\begin{rem}
The bound in Theorem \ref{th:BasicMSE} is meaningful only if $\sigma
_{\mathrm{as}}^{2}(P,f)
<\infty,$ $C_{0}(P)<\infty$, $C_{1}(P,f)< \infty$ and
$C_{2}(P,f)<\infty$.
Under Assumption \ref{as:SmallSetCond}, we always have
$\mathbb{E}_{\nu} T<\infty$ but not necessarily $\mathbb{E}_{\nu}
T^{2}<\infty$.
On the other hand, finiteness of $\mathbb{E}_{\nu} (\Xi({\bar
f}))^{2}$ is a
sufficient and necessary
condition for the CLT to hold for Markov chain $X_{n}$ and function
$f$. This fact is proved in \cite{ECP2008-9}
in a more general setting. For our purposes, it is important to note
that $\sigma_{\mathrm{as}}^{2}(P,f)$ in Theorem \ref{th:BasicMSE} is
indeed the \textit
{asymptotic variance} which
appears in the CLT, that is
\[
\sqrt{n} (\hat\theta_{n}-\theta)\to_{d}\mathrm{N}
\bigl(0,\sigma_{\mathrm{as}}^{2}(P,f)\bigr).
\]
Moreover,
\[
{\lim_{n\to\infty}}n\mathbb{E}_{\xi} (\hat\theta_{n}-\theta
)^2 =\sigma_{\mathrm{as}}^{2}(P,f).
\]
In this sense, the leading term ${\sigma_{\mathrm{as}}(P,f)}/{\sqrt
{n}}$ in Theorem
\ref{th:BasicMSE} is
``asymptotically correct'' and cannot be improved.
\end{rem}

\begin{rem}
Under additional assumptions of geometric and polynomial
ergodicity, in Sections~\ref{Sec:Geom} and~\ref{Sec:Poly} respectively, we will derive bounds
for $\sigma_{\mathrm{as}}^{2}(P,f)$ and $C_{0}(P),$ $C_{1}(P,f),$
$C_{2}(P,f)$ in terms of some
explicitly computable quantities.
\end{rem}

\begin{rem}
In our related work \cite{latuszynski2011nonasymptotic}, we discuss a
special case of the setting considered here, namely when regeneration
times $T_k$ are identifiable. These leads to $X_0 \sim\nu$ and an
regenerative estimator
of the form
%
\begin{eqnarray}
\label{eq:regen_est} \hat{\theta}_{T_{R(n)}} & :=& \frac{1}{
T_{R(n)} } \sum_{i=1}^{R(n)}
\Xi_i(f) = \frac{1}{ T_{R(n)} } \sum
_{i=0}^{T_{R(n)}-1}f(X_i).
\end{eqnarray}
The estimator $\hat{\theta}_{T_{R(n)}}$ is somewhat easier to
analyze. We refer to \cite{latuszynski2011nonasymptotic} for details.
\end{rem}

\begin{pf*}{Proof of Theorem \ref{th:BasicMSE}}
Recall $R(n)$ defined in \eqref{eqdef:Rn} and let
\[
\label{eq: R} \Delta(n) := T_{R(n)}-n.
\]
In words: $R(n)$ is \textit{the first moment of regeneration past $n$}
and $\Delta(n)$ is the \textit{overshoot}
or excess over $n$. Let us express the estimation error as follows.
\begin{eqnarray*}
  \hat\theta_{n}-\theta& =&
\frac{1}{n}\sum_{i=0}^{n-1}{\bar
f}(X_i) %
= \frac{1}{n} \Biggl(
\sum_{i=T_1}^{T_{R(n)}-1}{\bar f}(X_i) +
\sum_{i=0}^{T_1-1}{\bar f} (X_i)-
\sum_{i=n}^{T_{R(n)}-1}{\bar f}(X_i)
\Biggr)
\\
& =:& \frac{1}{n} (\mathcal{Z} +\mathcal{O}_{1} -\mathcal{O}_{2} ),
\end{eqnarray*}
with the convention that $\sum_l^u=0$ whenever $l>u$.
The triangle inequality entails
%
\begin{equation}
\label{eq:MSEinproof}
\sqrt{\mathbb{E}_{\xi} ( \hat
\theta_{n}-\theta)^{2
}} \leq\frac{1}{n}
\bigl( \sqrt{\mathbb{E}_{\xi}\mathcal{Z}^{2}}+\sqrt{
\mathbb{E}_{\xi}(\mathcal{O}_{1}-\mathcal{O}_{2})^{2}}
\bigr).
\end{equation}
Denote $C(P,f) := \sqrt{\mathbb{E}_{\xi}(\mathcal{O}_{1}-\mathcal{O}_{2})^{2}}$
and compute
%
%
\begin{eqnarray}\label{eq:Cf}
\nonumber
C(P,f)&=& \Biggl(\mathbb{E}_\xi\Biggl\{ \Biggl(\sum
_{i=0}^{T-1}{\bar f} (X_i)-\sum
_{i=n}^{T_{R(n)}-1}{\bar f}(X_i) \Biggr)\mathbb{I}(T
\geq n)
\\
\nonumber
& &\hspace*{27pt}{} + \Biggl(\sum_{i=0}^{T-1}
{\bar f}(X_i)-\sum_{i=n}^{T_{R(n)}-1}{
\bar f} (X_i) \Biggr)\mathbb{I}(T<n) \Biggr\}^2
\Biggr)^{1/2}
\\
&\leq& \Biggl(\mathbb{E}_\xi\Biggl(\sum
_{i=0}^{T-1}\big|{\bar f}(X_i)\big|+\sum
_{i=n}^{T_{R(n)}-1}\big|{\bar f}(X_i)\big|\mathbb{I}(T< n)
\Biggr)^2 \Biggr)^{1/
2}
\\
\nonumber
&\leq&\sqrt{\mathbb{E}_\xi\Biggl(\sum
_{i=0}^{T-1}\big|{\bar f}(X_i)\big|
\Biggr)^2}+\sqrt{\mathbb{E}_\xi\Biggl(\sum
_{i=n}^{T_{R(n)}-1}\big|{\bar f}(X_i)\big|\mathbb{I}(T< n)
\Biggr)^2}
\\
&= &C_{1}(P,f)+C_{2}(P,f).\nonumber
\end{eqnarray}
%
%
It remains to bound the middle term, $\mathbb{E}_{\xi}\mathcal{Z}^{2}$, which
clearly corresponds to the most
significant portion of the estimation error.
The crucial step in our proof is to show the following inequality:
%
\begin{equation}
\label{eq: Regenest} \mathbb{E}_{\nu} \Biggl( \sum
_{i=0}^{T_{R(n)}-1}{\bar f}(X_i)
\Biggr)^{2} \leq\sigma_{\mathrm{as}}^{2}(P,f)
\bigl(n+2C_{0}(P) \bigr).
\end{equation}
Once this is proved, it is easy to see that
\begin{eqnarray*}
\mathbb{E}_{\xi}\mathcal{Z}^{2}& =& \sum_{j=1}^{n}
\mathbb{E}_{\xi} \bigl( \mathcal{Z}^{2} |T_1=j \bigr)
\mathbb{P}_{\xi}(T_1=j)
= \sum
_{j=1}^{n} \mathbb{E}_{\nu} \Biggl( \sum
_{i=0}^{T_{R(n-j)}-1}{\bar f}(X_i)
\Biggr)^{2} \mathbb{P}_{\xi}(T_1=j)
\\
& \leq&\sum_{j=1}^{n}
\sigma_{\mathrm{as}}^{2}(P,f) \bigl(n-j+2C_{0}(P) \bigr)
\mathbb{P}_{\xi
}(T_1=j)
\leq
\sigma_{\mathrm{as}}^{2}(P,f) \bigl(n+2C_{0}(P) \bigr),
\end{eqnarray*}
consequently, $\sqrt{\mathbb{E}_{\xi}\mathcal{Z}^{2}}\leq\sqrt
{n}\sigma
_{\mathrm{as}}(P,f)
(1+C_{0}(P)/n)$ and the conclusion will follow by recalling \eqref
{eq:MSEinproof} and \eqref{eq:Cf}.

We are therefore left with the task of proving \eqref{eq: Regenest}.
This is essentially a statement about
sums of i.i.d. random variables. Indeed,
%
\begin{equation}
\label{eq: Blocks} \sum_{i=0}^{T_{R(n)}-1}{
\bar f}(X_i)=\sum_{k=1}^{R(n)}
\Xi_{k}({\bar f})
\end{equation}
and all the blocks $\Xi_k$ (including $\Xi=\Xi_1$) are i.i.d.  under
$\mathbb{P}_{\nu}$.
By the general version of the Kac theorem (\cite{meyn1993markov},
Theorem
10.0.1, or \cite{nummelin2002mc}, equation (3.3.7)), we have
\[
\mathbb{E}_{\nu} \Xi(f)= \pi(f) \mathbb{E}_{\nu} T
\]
(and $1/\mathbb{E}_{\nu} T=\beta\pi(J)$), so $\mathbb{E}_{\nu}
\Xi({\bar f})=0$
and $\operatorname{Var}_{\nu} \Xi({\bar f})=\sigma_{\mathrm
{as}}^{2}(P,f)\mathbb{E}
_{\nu} T$.
Now we will exploit the fact that $R(n)$ is a stopping time with
respect to
$\mathcal{G}_k=\sigma((\Xi_{1}({\bar f}),\tau_{1}),\ldots,(\Xi
_{k}({\bar
f}),\tau_{k}))$,
a~filtration generated by i.i.d. pairs.
We are in a position to apply the two Wald's identities. The second
identity yields
\[
\mathbb{E}_{\nu} \Biggl(\sum_{k=1}^{R(n)}
\Xi_k({\bar f}) \Biggr)^2=\operatorname{Var}_{\nu} \Xi({\bar
f}) \mathbb{E}_{\nu} R(n)= \sigma_{\mathrm{as}}^{2}(P,f)
\mathbb{E}_{\nu} T \mathbb{E}_{\nu} R(n).
\]
But in this expression, we can replace $\mathbb{E}_{\nu}T\mathbb
{E}_{\nu} R(n)$
by $\mathbb{E}_{\nu}T_{R(n)}$ because of the first
Wald's identity:
\[
\mathbb{E}_{\nu} T_{R(n)}= \mathbb{E}_{\nu} \sum
_{k=1}^{R(n)}\tau_k=\mathbb{E}_{\nu
}T
\mathbb{E}_{\nu} R(n).
\]
It follows that
%
\begin{equation}
\label{eq: Seqident} \mathbb{E}_{\nu} \Biggl(\sum
_{k=1}^{R(n)}\Xi_{k}({\bar f})
\Biggr)^{2} = \sigma_{\mathrm{as}}^{2}(P,f)
\mathbb{E}_{\nu}T_{R(n)}=\sigma_{\mathrm
{as}}^{2}(P,f)
\bigl(n+\mathbb{E}_{\nu} \Delta(n) \bigr).
\end{equation}

We now focus attention on bounding the ``mean overshoot'' $\mathbb
{E}_{\nu}
\Delta(n)$.
Under $\mathbb{P}_{\nu}$, the cumulative sums $T=T_1<T_2<\cdots
<T_k<\cdots
$ form a (nondelayed) renewal process
in discrete time.
%
Let us invoke the following elegant theorem of Lorden (\cite
{lorden1970excess}, Theorem 1):
%
\begin{equation}
\label{eq: Lord} \mathbb{E}_{\nu} \Delta(n) \leq\frac{ \mathbb
{E}_{\nu}
T^2}{\mathbb{E}_{\nu} T}.
\end{equation}
By Lemma~\ref{lem: SquareBlock} with $g \equiv1$ from Section~\ref
{Sec:Lemmas}, we obtain:
%
\begin{equation}
\label{eq: Lord2} \mathbb{E}_{\nu} \Delta(n) \leq2
\mathbb{E}_{\pi}T-1.
\end{equation}
%
%
%
Hence, substituting \eqref{eq: Lord2} into \eqref{eq: Seqident} and
taking into account
\eqref{eq: Blocks} we obtain \eqref{eq: Regenest} and complete the proof.
\end{pf*}

\section{Geometrically ergodic chains}\label{Sec:Geom}

In this section, we upper bound constants $ \sigma_{\mathrm
{as}}^{2}(P,f), C_{0}(P), C_{1}(P,f),
C_{2}(P,f),$ appearing in Theorem \ref{th:BasicMSE}, for geometrically
ergodic Markov chains under a quantitative drift assumption. Proofs are
deferred to Sections \ref{Sec:Lemmas} and \ref{Sec:Proofs}.

Using drift conditions is a standard approach for establishing
geometric ergodicity.
We refer to \cite{roberts2004general} or \cite{meyn1993markov} for
definitions and further details.
The assumption below is the same as in \cite{baxendale2005renewal}.
Specifically, let $J$ be the small set which appears in Assumption~\ref
{as:SmallSetCond}.

\begin{assumption}[(Geometric drift)]\label{as: GeoDriftCond}
There exist a function $V\dvtx \mathcal{X}\to[1,\infty[$, constants
$\lambda<1$ and
$K<\infty$ such that
\[
PV(x):=\int_{\mathcal{X}}P(x,\mathrm{d}y)V(y)\leq%
\cases{ \lambda V(x),& \quad for $x\notin J$,
\cr
K,&\quad  for $x\in J$.
}
\]
\end{assumption}

In many papers conditions similar to Assumption \ref{as: GeoDriftCond}
have been established for
realistic MCMC algorithms in statistical models of practical relevance
\cite
{hobert1998geometric,fort2000v,fort2003geometric,jones2004sufficient,johnson2010gibbs,roy2010monte}.
This opens the possibility of computing nonasymptotic upper bounds on
MSE or nonasymptotic confidence intervals in these models.

In this section, we bound quantities appearing in Theorem
\ref{th:BasicMSE} by expressions involving $\lambda$, $\beta$ and $K$.
The main result in this section is the following theorem.

\begin{theorem}\label{th: DriftBounds}
If Assumptions \ref{as:SmallSetCond} and \ref{as: GeoDriftCond} hold
and $f$ is such that
\[
\Vert{\bar f}\Vert_{V^{1/2}}:= \sup_{x}\big|{\bar
f}(x)\big|/V^{1/
2}(x)<\infty,
\]
then
\begin{eqnarray*}
&&\phantom{ii}\mathrm{(i)}\quad  C_{0}(P) \leq\frac{\lambda}{1-\lambda} \pi(V) +
\frac{K-\lambda-\beta}{\beta(1-\lambda)}+\frac{1}{2},
\\
&&\phantom{i}\mathrm{(ii)}\quad  \frac{ \sigma_{\mathrm{as}}^{2}(P,f)}{\Vert{\bar f}
\Vert_{V^{1/
2}}^{2}} \leq\frac{1
+\lambda^{1/2}}{1-\lambda^{1/2}}\pi(V)+
\frac{2(K^{1/2}-\lambda^{1/2}-\beta)}{\beta(1-\lambda^{1/2})}\pi\bigl
(V^{1/2}\bigr),
\\
&&\hspace*{1pt}\mathrm{(iii)}\quad \frac{C_{1}(P,f)^2}{\Vert{\bar f}
\Vert_{V^{1/2}}^{2} } \leq\frac{1}{(1-\lambda^{1/2})^{2}}\xi( V) +
\frac{2(K^{1/2}-\lambda^{1/2}-\beta)}{\beta(1-\lambda
^{1/2})^{2}} \xi\bigl(V^{1/2}\bigr)
\\
&&\hspace*{82pt}{} + \frac{\beta(K-\lambda-\beta)+2(K^{1/2}-\lambda^{1/2}-\beta)^{2}} {
\beta^{2}(1-\lambda^{1/2})^{2}}, 
\\
&&\hspace*{2pt}\mathrm{(iv)}\quad
\begin{tabular}{@{}p{300pt}@{}}
$C_{2}(P,f)^2$ satisfies an inequality
analogous to \textup{(iii)} with $\xi$ replaced by $\xi P^n$.
\end{tabular}
\end{eqnarray*}
\end{theorem}

\begin{rem}
Combining Theorem \ref{th: DriftBounds} with Theorem \ref
{th:BasicMSE} yields the MSE bound of interest. Note that the leading
term is of order $n^{-1}\beta^{-1}(1-\lambda)^{-1}.$ A related result
is Proposition 2 of \cite{fort2003convergence}
where the $p$th moment of $\hat{\theta}_n$ for $p \geq2$ is
controlled under similar assumptions.
Specialised to $p=2$, the leading term of the moment bound of \cite
{fort2003convergence} is of order $n^{-1}\beta^{-3}(1-\lambda)^{-4}.$
\end{rem}

\begin{rem}
An alternative form of the first bound in Theorem \ref{th:
DriftBounds} is
\[
\mathrm{(i')}\quad  C_{0}(P) \leq
\frac{\lambda^{1/2} }{1-\lambda^{1/2}} \pi\bigl(V^{1/2}\bigr) +\frac{K^{1/2}-\lambda
^{1/2}-\beta}{\beta(1-\lambda^{1/2})}+
\frac{1}{2}.
\]

Theorem \ref{th: DriftBounds} still involves some quantities which can
be difficult to compute, such as
$\pi(V^{1/2})$ and $\pi(V)$, not to mention $\xi P^n(V^{1/2})$ and $\xi P^n(V)$. The following proposition
gives some simple complementary bounds.
\end{rem}

\begin{prop}\label{pr: Compl} Under Assumptions \ref{as:SmallSetCond}
and \ref{as: GeoDriftCond},
\begin{eqnarray*}
&&\phantom{ii}\mathrm{(i)}\quad \pi\bigl(V^{1/2}\bigr) \leq\pi(J) \frac{K^{1/2}
-\lambda^{1/2}}{1-\lambda^{1/2}}
\leq\frac{K^{1/2} -
\lambda^{1/2}}{1-\lambda^{1/2}},
\\
&&\phantom{i}\mathrm{(ii)}\quad \pi(V) \leq\pi(J) \frac{K-\lambda}{1-\lambda} \leq\frac
{K-\lambda}{1-\lambda},
\\
&&\hspace*{1pt}\mathrm{(iii)}\quad \mbox{if } \xi\bigl(V^{1/2}\bigr) \leq\frac{K^{1/2}}{1-\lambda^{1/2}}
\qquad \mbox{then }  \xi P^n\bigl(V^{1/2}\bigr) \leq
\frac{K^{1/2}}{1-\lambda^{1/2}},
\\
&&\hspace*{2pt}\mathrm{(iv)}\quad \mbox{if }\xi(V) \leq\frac{K}{1-\lambda}  \qquad \mbox{then }  \xi
P^n(V) \leq\frac{K}{1-\lambda},
\\
&&\hspace*{1pt}\phantom{i}\mathrm{(v)}\quad \Vert{\bar f}\Vert_{V^{1/2}} \mbox{ can be related to }
\Vert f
\Vert_{V^{1/2}} \mbox{ by}
\\
&&\hspace*{2pt}\phantom{\mathrm{(iii)}\quad  } \Vert{\bar f}\Vert_{V^{1/2}} \leq\Vert f \Vert_{V^{1/2}}
\biggl[1 + \frac{\pi(J)(K^{1/2}-\lambda^{1/2})}{(1-\lambda^{1/2})
\inf_{x \in\mathcal{X}}V^{1/2}(x)} \biggr] \leq\Vert f \Vert
_{V^{1/2}}
\biggl[1 + \frac{K^{1/2}-\lambda^{1/2}}{1-\lambda^{1/2}} \biggr].
\end{eqnarray*}
\end{prop}

\begin{rem}
In MCMC practice, almost always the initial state is deterministically
chosen, \mbox{$\xi=\delta_{x}$} for some
$x\in\mathcal{X}$. In this case in (ii) and (iii), we just have
to choose
$x$ such that
$V^{1/2}(x)\leq K^{1/2}/(1-\lambda^{1/2})$ and
$V(x)\leq K/(1-\lambda)$, respectively (note that the latter inequality
implies the former). It might be interesting to note that our bounds
would not be improved if we added
a burn-in time $t>0$ at the beginning of simulation. The standard
practice in MCMC computations is to discard
the initial part of trajectory and use the estimator
\[
\hat\theta_{t,n} := \frac{1}{n}\sum
_{i=t}^{n+t-1}f(X_{i}).
\]
Heuristic justification is that the closer $\xi P^t$ is to the
equilibrium distribution $\pi$, the better.
However, for technical reasons, our upper bounds on error are the
tightest if the initial point has the smallest value of $V$,
and \textit{not} if its distribution is close to $\pi$.
\end{rem}

\begin{rem}
In many specific examples, one can obtain (with some additional effort)
sharper inequalities
than those in Proposition \ref{pr: Compl} or at least bound $\pi(J)$
away from~1.
However, in general we assume that such bounds are not available.
\end{rem}

\section{Polynomially ergodic Markov chains} \label{Sec:Poly}

In this section, we upper bound constants $ \sigma_{\mathrm
{as}}^{2}(P,f), C_{0}(P), C_{1}(P,f),
C_{2}(P,f),$ appearing in Theorem \ref{th:BasicMSE}, for polynomially
ergodic Markov chains under a quantitative drift assumption. Proofs are
deferred to Sections \ref{Sec:Lemmas} and \ref{Sec:Proofs}.

The following drift condition is a counterpart of Drift in Assumption
\ref{as: GeoDriftCond}, and is used to establish polynomial ergodicity
of Markov chains \cite{jarner2002polynomial,douc2004practical,douc2008bounds,meyn1993markov}.

\begin{assumption}[(Polynomial drift)]\label{as:DriftPoly}
There exist a function $V\dvtx \mathcal{X}\to[1,\infty[$, constants
$\lambda<1,\allowbreak
\alpha\leq1 $ and
$K<\infty$ such that
\[
PV(x)\leq%
\cases{ V(x)-(1-\lambda)V(x)^\alpha,&\quad
for $x\notin J$,
\cr
K,& \quad for $x\in J$.
}
\]
\end{assumption}

We note that Assumption \ref{as:DriftPoly} or closely related drift
conditions have been established for MCMC samplers in specific models
used in Bayesian inference, including independence samplers,
random-walk Metropolis algorithms, Langevin algorithms and Gibbs
samplers, see, for example, \cite{fort2000v,jarner2003necessary,jarner2007convergence}.

In this section, we bound quantities appearing in Theorem
\ref{th:BasicMSE} by expressions involving $\lambda$, $\beta$,
$\alpha$ and $K$.
The main result in this section is the following theorem.

\begin{theorem}\label{th: PolyDriftBounds}
If Assumptions \ref{as:SmallSetCond} and \ref{as:DriftPoly} hold with
$\alpha>\frac{2}{3}$ and $f$ is such that
$\Vert{\bar f}\Vert_{V^{{(3/2)}\alpha-1}}:= \sup_{x}|{\bar f}
(x)|/V^{{(3/2)}\alpha-1}(x)<\infty$, then
\begin{eqnarray*}
&&\hspace*{-5pt}\phantom{ii}\mathrm{(i)}\quad  C_{0}(P) \leq\frac{1}{\alpha(1-\lambda)} \pi
\bigl(V^\alpha\bigr) +\frac{K^\alpha-1 - \beta}{\beta\alpha
(1-\lambda)}+\frac{1}{\beta}-
\frac{1}{2},
\\
&&\hspace*{-5pt}\phantom{i}\mathrm{(ii)}\quad  \frac{ \sigma_{\mathrm{as}}^{2}(P,f)} {\Vert{\bar
f}\Vert_{V^{{(3/2)}\alpha-1}}^{2}} \leq\pi\bigl(V^{3\alpha
-2}\bigr)+
\frac{4 \pi(V^{2\alpha
-1})}{\alpha(1-\lambda)}
+2 \biggl(
\frac{2K^{{\alpha}/{2}}-2 - 2\beta}{\alpha\beta
(1-\lambda)}+\frac{1}{\beta}-1 \biggr)\pi\bigl(V^{{(3/2)}\alpha-1}
\bigr),%
\\
&&\hspace*{-5pt}\hspace*{1pt}\mathrm{(iii)}\quad  \frac{C_{1}(P,f)^2}{\Vert{\bar f}
\Vert_{V^{{(3/2)}\alpha-1}}^{2}} \leq\frac{1}{(2\alpha
-1)(1-\lambda)}\xi
\bigl(V^{2\alpha-1}\bigr)+ \frac{4}{\alpha^2(1-\lambda)^2}\xi
\bigl(V^{\alpha}\bigr)
\\
&&\hspace*{-5pt}\hspace*{91pt}{} + \biggl(\frac{8K^{\alpha/2}-8 - 8 \beta}{\alpha^2\beta
(1-\lambda)^2}+\frac{4-4\beta}{\alpha\beta(1-\lambda
)} \biggr)\xi
\bigl(V^{\alpha/2}\bigr)\\
&&\hspace*{-5pt}\hspace*{91pt}{} +\frac{\alpha(1-\lambda
)+4}{\alpha\beta(1-\lambda)} + \frac{K^{2\alpha-1}-1 - \beta
}{(2\alpha-1)\beta(1-\lambda)}
\\
&&\hspace*{-5pt}\hspace*{91pt}{}+\frac{4(K^{\alpha}-1 -\beta)}{\alpha^2\beta(1-\lambda)^2}+2
\biggl(\frac{2K^{{\alpha}/{2}}-2 - 2\beta}{\alpha\beta
(1-\lambda
)}+\frac{1}{
\beta} \biggr)^2\\
&&\hspace*{-5pt}\hspace*{91pt}{} -2 \biggl(\frac{2K^{{\alpha}/{2}}-2
- 2 \beta}{\alpha\beta(1-\lambda)}+\frac{1}{\beta} \biggr),
\\
&&\hspace*{-5pt}\hspace*{2pt}\mathrm{(iv)}\quad \frac {C_{2}(P,f)^2}{\Vert{\bar f}
\Vert_{V^{{(3/2)}\alpha-1}}^{2}} \leq\frac{1}{(2\alpha-1)\beta
^{(2\alpha-1)/\alpha
}(1-\lambda)} \biggl(\frac{K-
\lambda}{1-\lambda} \biggr)^{(4\alpha
-2)/\alpha} 
+ \frac{4(K-\lambda)^2}{\alpha^2\beta(1-\lambda)^4}
\\
&&\hspace*{-5pt}\hspace*{91pt}{} + \biggl(\frac{8K^{\alpha/2}-8 - 8 \beta}{\alpha^2\beta
(1-\lambda)^2}+\frac{4-4\beta}{\alpha\beta(1-\lambda
)} \biggr)\frac{K-\lambda
}{\sqrt{\beta}(1-\lambda)}\\
&&\hspace*{-5pt}\hspace*{91pt}{} +\frac{\alpha
(1-\lambda)+4}{\alpha\beta(1-\lambda)} + \frac{K^{2\alpha-1}-1 -
\beta}{(2\alpha-1)\beta(1-\lambda)}
\\
&&\hspace*{-5pt}\hspace*{91pt}{}+\frac{4(K^{\alpha}-1 -\beta)}{\alpha^2\beta(1-\lambda)^2}+2
\biggl(\frac{2K^{{\alpha}/{2}}-2 - 2\beta}{\alpha\beta
(1-\lambda
)}+\frac{1}{
\beta} \biggr)^2 \\
&&\hspace*{-5pt}\hspace*{91pt}{}-2 \biggl(\frac{2K^{{\alpha}/{2}}-2
- 2 \beta}{\alpha\beta(1-\lambda)}+\frac{1}{\beta} \biggr).
\end{eqnarray*}
\end{theorem}

\begin{rem}
A counterpart of Theorem~\ref{th: PolyDriftBounds} parts
(i)--(iii) for $\frac{1}{2}< \alpha\leq\frac{2}{3}$ and functions s.t.
$\Vert f \Vert_{V^{\alpha-{1}/{2}}}<\infty$ can be also
established, using respectively modified but analogous calculations as
in the proof of the above. For part (iv) however, an additional
assumption $\pi(V) < \infty$ is necessary.
\end{rem}

Theorem \ref{th: PolyDriftBounds} still involves some quantities
depending on $\pi$ which can be difficult to compute, such as
$\pi(V^\eta)$ for $\eta\leq\alpha$. The following proposition
gives some simple complementary bounds.

\begin{prop}\label{pr: PolyCompl} Under Assumptions \ref
{as:SmallSetCond} and \ref{as:DriftPoly},
\begin{enumerate}[(ii)]
\item[(i)] For $\eta\leq\alpha$ we have
\[
\pi\bigl(V^{\eta}\bigr) \leq\biggl(\frac{K - \lambda
}{1-\lambda
}
\biggr)^{\eta/\alpha}.
\]
\item[(ii)] If $\eta\leq\alpha$, then
$\Vert{\bar f}\Vert_{V^\eta} $ can be related to $ \Vert f \Vert
_{V^\eta} $ by
\[
\Vert{\bar f}\Vert_{V^\eta} \leq\Vert f
\Vert_{V^\eta} \biggl[1 + \biggl(\frac{K - \lambda}{1-\lambda}
\biggr)^{\eta/\alpha}
\biggr].
\]
\end{enumerate}
\end{prop}


\section{Applicability in Bayesian inference and examples} \label{Sec:Ex}

To apply current results for computing MSE of estimates arising in
Bayesian inference, one needs drift and small set conditions with
explicit constants. The quality of these constants will affect the
tightness of the overall MSE bound. In this section, we present three
numerical examples. In Section~\ref{subsec:hm}, a simplified
hierarchical model similar as \cite{jones2001honest}, Example 2, is
designed to compare the bounds with actual values and asses their
quality. Next, in Section~\ref{subsec:pg}, we upperbound the MSE in
the extensively discussed in literature Poisson--Gamma hierarchical
model. Finally, in Section~\ref{subsec:cn}, we present the
contracting normals toy-example to demonstrate numerical improvements
over~\cite{eps_alpha}.

In realistic statistical models, the explicit drift conditions required
for our analysis are very difficult to establish. Nevertheless, they
have been recently obtained for a wide range of complex models of
practical interest. Particular examples include: Gibbs sampling for
hierarchical random effects models in \cite{jones2004sufficient}; van
Dyk and Meng's algorithm for multivariate Student's $t$ model \cite{marchev2004geometric}; Gibbs sampling for a family of Bayesian
hierarchical general linear models in \cite{johnson2007gibbs} (cf.
also~\cite{johnson2010gibbs}); block Gibbs sampling for Bayesian
random effects models with improper priors \cite{tan2009block}; Data
Augmentation algorithm for Bayesian multivariate regression models with
Student's $t$ regression errors \cite{roy2010monte}. Moreover, a large
body of related work has been devoted to establishing a drift condition
together with a small set to enable regenerative simulation for classes
of statistical models. This kind of results, pursued in a number of
papers mainly by James~P.~Hobert, Galin~L.~Jones and their coauthors,
cannot be used directly for our purposes, but may provide substantial
help in establishing quantitative drift and regeneration required here.

In settings where existence of drift conditions can be established, but
explicit constants can not be computed (cf., e.g., \cite
{fort2003geometric,papaspiliopoulos2008stability}), our results do not
apply and one must validate MCMC by asymptotic arguments. This is not
surprising since qualitative existence results are not well suited for
deriving quantitative finite sample conclusions.\vspace*{-3pt}

\subsection{A simplified hierarchical model} \label{subsec:hm}
The simulation experiments described below are designed to compare
the bounds proved in this paper with actual errors of MCMC estimation.
We use a simple example similar as \cite{jones2001honest}, Example~2.
Assume that $y=(y_{1},\ldots,y_{t})$ is an i.i.d. sample from the
normal distribution $\mathrm{N}(\mu,\kappa^{-1})$,
where $\kappa$ denotes the reciprocal of the variance. Thus, we have
\[
p(y|\mu,\kappa) = p(y_{1},\ldots,y_{t}|
\mu,\kappa) \propto{\kappa^{t/2}}\exp\Biggl[-
\frac{\kappa}{2} \sum_{j=1}^t(y_{j}-
\mu)^{2} \Biggr].
\]
The pair $(\mu,\kappa)$ plays the role of an unknown parameter. To
make things simple, let us use the Jeffrey's noninformative (improper) prior
$p(\mu,\kappa)=p(\mu)p(\kappa)\propto\kappa^{-1}$ (in \cite
{jones2001honest} a different prior is considered).
The posterior density is
\[
  p(\mu,\kappa|y) \propto p(y|\mu,
\kappa)p(\mu,\kappa) \propto\kappa^{t/2-1}\exp\biggl[-
\frac{\kappa t}{2} \bigl(s^{2}+(\bar y-\mu)^{2} \bigr)
\biggr],
\]
where
\[
\bar y=\frac{1}{t}\sum_{j=1}^{t}y_{j},\qquad
s^2=\frac{1}{t}\sum_{j=1}^{t}(y_{j}-
\bar y)^{2}.
\]
Note that $\bar y$ and $s^{2}$ only determine the location and scale of
the posterior.
We will be using a Gibbs sampler, whose performance does not depend on
scale and location,
therefore without loss of generality we can assume that $\bar y=0$ and
$s^{2}=t$.
Since $y=(y_{1},\ldots,y_{t})$ is kept fixed, let us slightly abuse notation
by using symbols $p(\kappa|\mu)$, $p(\mu|\kappa)$ and $p(\mu)$ for
$p(\kappa|\mu,y)$, $p(\mu|\kappa,y)$ and $p(\mu|y)$, respectively.
The Gibbs sampler alternates between drawing samples from both conditionals.
Start with some $(\mu_{0},\kappa_{0})$. Then, for $i=1,2,\ldots,$
\begin{itemize}
\item$\kappa_{i}\sim\operatorname{Gamma}({t}/{2},({t}/{2})(s^{2}+\mu
_{i-1}^{2}) )$,
\item$\mu_{i}\sim\mathrm{N}(0,1/(\kappa_{i}t) )$.
\end{itemize}
If we are chiefly interested in $\mu$, then it is convenient to
consider the two small steps
$\mu_{i-1}\to\kappa_{i}\to\mu_{i}$ together. The transition
density is
\begin{eqnarray*}
  p(\mu_{i}|\mu_{i-1})&=&
\int p(\mu_{i}|\kappa) p(\kappa|\mu_{i-1})\, \mathrm{d}\kappa
\\
&\propto&\int_{0}^{\infty} \kappa^{1/2}\exp
\biggl[-\frac{\kappa t}{2}\mu_{i}^2 \biggr]
\bigl(s^{2}+\mu_{i-1}^{2} \bigr)^{t/2}
\kappa^{t/2-1} \exp\biggl[-\frac{\kappa t}{2} \bigl(s^{2}+
\mu_{i-1}^{2} \bigr) \biggr] \,\mathrm{d}\kappa
\\
&=& \bigl(s^{2}+\mu_{i-1}^{2}
\bigr)^{t/2} \int_{0}^{\infty}
\kappa^{(t-1)/2} \exp\biggl[-\frac{\kappa t}{2} \bigl(s^{2}+
\mu_{i-1}^{2}+\mu_{i}^{2} \bigr)
\biggr] \,\mathrm{d}\kappa
\\
&\propto&\bigl(s^{2}+\mu_{i-1}^{2}
\bigr)^{t/2} \bigl(s^{2}+\mu_{i-1}^{2}+
\mu_{i}^{2} \bigr)^{-(t+1)/2}.
\end{eqnarray*}
The proportionality constants concealed behind the $\propto$ sign
depend only on $t$.
Finally, we fix scale letting $s^{2}=t$ and get
%
\begin{equation}
\label{eq: trans} p(\mu_{i}|\mu_{i-1})\propto
\biggl(1+\frac{\mu_{i-1}^{2}}{t} \biggr)^{t/2} \biggl(1+\frac{\mu
_{i-1}^{2}}{t}+
\frac{\mu_{i}^{2}}{t} \biggr)^{-(t+1)/2}.
\end{equation}
If we consider the RHS of \eqref{eq: trans} as a function of $\mu
_{i}$ only, we can regard the first factor as constant and write
\[
p(\mu_{i}|\mu_{i-1}) \propto\biggl(1+ \biggl(1+
\frac{\mu_{i-1}^{2}}{t} \biggr)^{-1}\frac
{\mu_{i}^{2}}{t} \biggr)^{-(t+1)/2}.
\]
It is clear that the conditional distribution of random variable
%
\begin{equation}
\label{eq: t} \mu_{i} \biggl(1+\frac{\mu_{i-1}^{2}}{t}
\biggr)^{-1/2}
\end{equation}
is t-Student distribution with $t$ degrees of freedom.
Therefore, since the t-distribution has the second moment equal to
$t/(t-2)$ for $t>2$, we infer that
\[
\mathbb{E}\bigl(\mu_{i}^{2}|\mu_{i-1}\bigr) =
\frac{t+\mu_{i-1}^{2}}{t-2}.
\]
Similar computation shows that the posterior marginal density of $\mu$
satisfies
\[
p(\mu)\propto\biggl(1+\frac{t-1}{t}\frac{\mu^{2}}{t-1}
\biggr)^{-t/2}.
\]
Thus, the stationary distribution of our Gibbs sampler is rescaled
t-Student with $t-1$ degrees of freedom. Consequently, we have
\[
\mathbb{E}_{\pi} \mu^2=\frac{t}{t-3}.
\]

\begin{prop}[(Drift)]\label{pr: Driftex}
Assume that $t\geq4$. Let $ V(\mu):=\mu^2+1 $
and $J=[-a,a]$. The transition kernel of the (2-step) Gibbs sampler
satisfies
\[
PV (\mu) \leq%
\cases{ \lambda V(\mu), &  \quad for $ |
\mu|>a$;
\cr
K, & \quad  for $|\mu|\leq a$,
}\qquad
\mbox{provided that }a>\sqrt{t/(t-3)}.
\]
The quantities $\lambda,$ $K$ and $\pi(V)$
are given by
\[
  \lambda=\frac{1}{t-2} \biggl(
\frac{2t-3}{1+a^{2}}+1 \biggr),\qquad  K=2+\frac{a^{2}+2}{t-2}\quad  \mbox{and}\quad
\pi(V) =\frac{2t-3}{t-3}.
\]
%
\end{prop}


\begin{pf} Since $a>\sqrt{t/t-3}$, we obtain that $\lambda=\frac
{1}{t-2} (\frac{2t-3}{1+a^2}+1 )<\frac{1}{t-2}(t-2)=1$.
Using the fact that
\[
PV(\mu)=\mathbb{E}\bigl(\mu_{i}^{2}+1|
\mu_{i-1}=\mu\bigr)= \frac{t+\mu^{2}}{t-2}+1
\]
we obtain
\begin{eqnarray*}
\lambda V(\mu)-PV(\mu)&=&\frac{1}{t-2} \biggl(\frac
{2t-3}{1+a^2}+1 \biggr)
\bigl(\mu^2+1\bigr)-\frac{t+\mu^2}{t-2}-1
\\
&=&\frac{1}{t-2} \biggl(\frac{2t-3}{1+a^2}\mu^2 +
\frac
{2t-3}{1+a^2}-2t+3 \biggr)
\\
&=&\frac{2t-3}{(t-2)(1+a^2)} \bigl(\mu^2+1-1-a^2 \bigr)
\\
&=&\frac{2t-3}{(t-2)(1+a^2)}\bigl(\mu^2-a^2\bigr).
\end{eqnarray*}
Hence, $\lambda V(\mu)-PV(\mu)>0$ for $|\mu|>a$. For $\mu$ such
that $|\mu|\leq a$, we get that
\[
PV(\mu)=\frac{t+\mu^{2}}{t-2}+1\leq\frac{t+a^2}{t-2}+1=2+\frac
{t+a^2-t+2}{t-2}=2+
\frac{a^2+2}{t-2}.
\]
Finally,
\[
\pi(V)= \mathbb{E}_{\pi} \mu^2+1 =\frac{t}{t-3}+1=
\frac{2t-3}{t-3}.
\]
\upqed\end{pf}

\begin{prop}[(Minorization)]\label{pr: Minex}
Let $p_{\min}$ be a subprobability density given by
\[
p_{\min}(\mu)= %
\cases{ p(\mu|a), & \quad  for $
|\mu|\leq h(a)$;
\cr
p(\mu|0), & \quad  for $|\mu|> h(a)$,
}
\]
where $p(\cdot|\cdot)$ is the transition density given by \eqref{eq:
trans} and
\[
h(a)= \biggl\{a^{2} \biggl[ \biggl(1+\frac{a^{2}}{t}
\biggr)^{t/(t+1)}-1 \biggr]^{-1}-t \biggr\}^{1/2}.
\]
Then $|\mu_{i-1}|\leq a$ implies $p(\mu_{i}|\mu_{i-1})\geq p_{\min}
(\mu_{i})$. Consequently,
if we take for $\nu$ the probability measure with the normalized
density $p_{\min}/\beta$
then the small set Assumption \ref{as:SmallSetCond} holds for \mbox{$J=[-a,a]$}.
Constant $\beta$ is given by
\[
\beta=1-\mathbb{P}\bigl(|\vartheta|\leq h(a) \bigr) +\mathbb
{P}\biggl
(|\vartheta|\leq
\biggl(1+\frac{a^{2}}{t} \biggr)^{-1/2} h(a) \biggr),
\]
%
%
\begin{figure}

\includegraphics{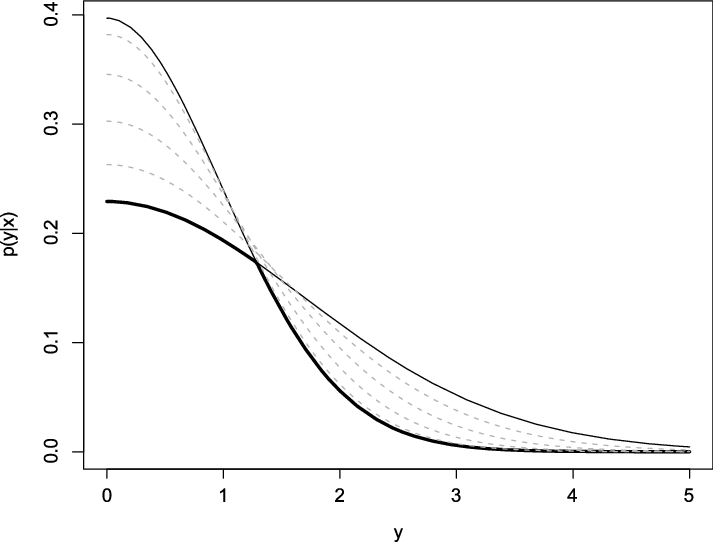}

\caption{Illustration of Proposition \protect\ref{pr: Minex}, with $t=50$ and $a=10$.
Solid lines are graphs of
$p(\mu_i|0)$ and $p(\mu_i|a)$. Bold line is the graph of $p_{\min}(\mu_i)$. Gray dotted lines are graphs of
$p(\mu_i|\mu_{i-1})$ for some selected positive $\mu_{i-1}\leq a$.}\label{fig1}
\end{figure}
where $\vartheta$ is a random variable with t-Student distribution
with $t$ degrees of freedom.
\end{prop}


Proposition \ref{pr: Minex} is illustrated in Figure \ref{fig1}.

\begin{pf*}{Proof of Proposition \ref{pr: Minex}}
The formula for $p_{\min}$ results from minimisation of $p(\mu
_{i}|\mu_{i-1})$ with respect
to $\mu_{i-1}\in[-a,a]$. We use \eqref{eq: trans}.
First, compute $(\partial/\partial\mu_{i-1}) p(\mu_{i}|\mu_{i-1})$
to check that for every $\mu_{i}$ the function
$\mu_{i-1}\mapsto p(\mu_{i}|\mu_{i-1})$ has to attain minimum either
at 0 or at $a$.
Indeed,
\begin{eqnarray*}
  \frac{\partial}{\partial\mu_{i-1}}p(\mu_{i}|
\mu_{i-1}) & =& \mathit{const}\cdot\biggl[\frac{t}{2}
\bigl(s^2+\mu_{i-1}^2\bigr)^{t/2-1}
\bigl(s^2+\mu_{i-1}^2+\mu_{i}^2
\bigr)^{-(t+1)/2}\cdot2\mu_{i-1}
\\
&&\hspace*{32pt}{} -\frac{t+1}{2}\bigl(s^2+
\mu_{i-1}^2\bigr)^{t/2}\bigl(s^2+
\mu_{i-1}^2+\mu_{i}^2
\bigr)^{-(t+1)/2-1}\cdot2\mu_{i-1} \biggr]
\\
& =& \mu_{i-1}\bigl(s^2+\mu_{i-1}^2
\bigr)^{t/2-1}\bigl(s^2+\mu_{i-1}^2+
\mu_{i}^2\bigr)^{-(t+1)/2-1}
\\
&&{} \cdot\bigl[t\bigl(s^2+\mu_{i-1}^2+
\mu_{i}^2\bigr)-(t+1) \bigl(s^2+
\mu_{i-1}^2+\mu_{i}^2\bigr) \bigr].
\end{eqnarray*}
Assuming that $\mu_{i-1}>0$, the first factor at the right-hand side
of the above equation is positive, so $(\partial/\partial\mu
_{i-1})p(\mu_{i}|\mu_{i-1})>0$ iff
$t(s^2+\mu_{i-1}^2+\mu_{i}^2)-(t+1)(s^2+\mu_{i-1}^2+\mu_{i}^2)>0$,
that is iff
\[
\mu_{i-1}^2 < t\mu_{i}^2-s^2.
\]
Consequently, if $t\mu_{i}^2-s^2\leq0$ then the function $\mu
_{i-1}\mapsto p(\mu_{i}|\mu_{i-1})$ is decreasing for $\mu_{i-1}>0$
and $\min_{0\leq\mu_{i-1}\leq a}p(\mu_{i}|\mu_{i-1})=p(\mu_{i},a)$.
If $t\mu_{i}^2-s^2> 0$, then this function first increases and then
decreases. In either case we have $\min_{0\leq\mu_{i-1}\leq a}p(\mu
_{i}|\mu_{i-1})=
\min[p(\mu_{i}|a),p(\mu_{i}|0)]$. Thus using symmetry, $p(\mu
_{i}|\mu_{i-1})=p(\mu_{i}|-\mu_{i-1})$, we obtain
\[
p_{\min}(\mu_i)=\min_{|\mu_{i-1}|\leq a}p(
\mu_{i}|\mu_{i-1})= %
\cases{ p(
\mu_i|a), & \quad  if $p(\mu_i|a) \leq p(
\mu_i|0) $;
\cr
p(\mu_i|0), &\quad  if $p(\mu_i|a) > p(
\mu_i|0)$.
}
\]
Now it is enough to solve the inequality, say, $p(\mu|0)< p(\mu|a)$,
with respect to
$\mu$. The following elementary computation shows that this inequality
is fulfilled iff $|\mu|> h(a)$:
\begin{eqnarray*}
  p(\mu|0) &=& \frac{(s^2)^{t/2}}{(s^2+\mu^2)^{(t+1)/2}}
< \frac{(s^2+a^2)^{t/2}}{(s^2+a^2+\mu^2)^{(t+1)/2}} = p(\mu|a),
\qquad \mbox{iff }
\\
\biggl(\frac{s^2+a^2+\mu^2}{s^2+\mu^2} \biggr)^{(t+1)/2}& <& \biggl(
\frac{s^2+a^2}{s^2} \biggr)^{t/2},\qquad  \mbox{iff }
\\
\biggl(1+\frac{a^2}{s^2+\mu^2} \biggr)^{t+1}& <& \biggl(1+
\frac{a^2}{s^2} \biggr)^{t},\qquad  \mbox{iff }
\\
\frac{a^2}{s^2+\mu^2}& < &\biggl(1+\frac{a^2}{s^2}
\biggr)^{t/(t+1)}-1,\qquad  \mbox{iff }
\\
\mu^2& >& a^2 \biggl[ \biggl(1+
\frac{a^2}{s^2} \biggr)^{t/(t+1)}-1 \biggr]^{-1}-s^2.
\end{eqnarray*}
It is enough to recall that $s^2=t$ and thus the right-hand side above
is just $h(a)^2$.

To obtain the formula for $\beta$, note that
\[
\beta=\int p_{\min}(\mu)\,\mathrm{d}\mu= \int_{|\mu|\leq h(a)}
p(\mu|a) \,\mathrm{d}\mu+\int_{|\mu|> h(a)} p(\mu|0)\,\mathrm{d}\mu
\]
and use \eqref{eq: t}.
\end{pf*}

\begin{rem}
It is interesting to compare the asymptotic behaviour of the constants in
Propositions \ref{pr: Driftex} and \ref{pr: Minex} for $a\to\infty$.
We can immediately see that $\lambda^{2}\to1/(t-2)$ and $K^{2}\sim
a^{2}/(t-2)$. Slightly more
tedious computation reveals that $h(a)\sim{\mathit{const}}\cdot
a^{1/(t+1)}$
and consequently
$\beta\sim{\mathit{const}}\cdot a^{-t/(t+1)}$.
\end{rem}


The parameter of interest is the posterior mean (Bayes estimator of
$\mu$). Thus, we
let $f(\mu)=\mu$ and $\theta=\mathbb{E}_{\pi}\mu=0$. Note that
our chain
$\mu_{0},\ldots,\mu_{i},\ldots$
is a sequence of martingale differences, so ${\bar f}=f$ and
\[
\sigma_{\mathrm{as}}^{2}(P,f)=\mathbb{E}_{\pi}
\bigl(f^{2}\bigr)=\frac{t}{t-3}.
\]
The MSE of the estimator $\hat\theta_{n}=\sum_{i=0}^{n-1} \mu_n$
can be also
expressed analytically, namely
\[
  \mathrm{MSE}=\mathbb{E}_{\mu_0}\hat
\theta_{n}^2=\frac{t}{n(t-3)}-\frac
{t(t-2)}{n^2(t-3)^2}
\biggl[1- \biggl(\frac{1}{t-2} \biggr)^n \biggr] +
\frac{t-2}{n^2(t-3)} \biggl[1- \biggl(\frac{1}{t-2} \biggr)^n
\biggr]\mu_0^2.
\]
Obviously, we have $\Vert f\Vert_{V^{1/2}} =1$.

We now proceed to examine the bounds proved in Section \ref{Sec:Geom}
under the geometric drift condition, Assumption \ref{as:
GeoDriftCond}. Inequalities for the asymptotic variance
play the crucial role in our approach. Let us fix $t=50$. Figure \ref{fig2}
shows how our bounds on $\sigma_{\mathrm{as}}(P,f)$
depend on the choice of the small set $J=[-a,a]$.

The gray solid line gives the bound
of Theorem \ref{th: DriftBounds}(ii) which assumes the knowledge of
$\pi V$ (and uses the obvious inequality
$\pi( V^{1/2})\leq(\pi V)^{1/2})$.
The black dashed line corresponds to a bound which involves only
$\lambda$, $K$ and $\beta$. It is
obtained if values of $\pi V$ and $\pi V^{1/2}$ are replaced by their
respective bounds given in
Proposition \ref{pr: Compl}(i) and (ii).

The best values of the bounds, equal to $2.68$ and $2.38$, correspond
to $a=3.91$ and $a=4.30$, respectively.
The actual value of the root asymptotic variance is $\sigma_{\mathrm
{as}}(P,f)=1.031$.
In Table \ref{tab1} below, we summarise
the analogous bounds for three values of $t$.

The results obtained for different values of parameter $t$ lead to
qualitatively similar conclusions. From now on, we keep $t=50$ fixed.

Table \ref{tab2} is analogous to Table \ref{tab1} but focuses on other constants
introduced in Theorem \ref{th:BasicMSE}. Apart from $\sigma_{\mathrm
{as}}(P,f)$, we compare
$C_{0}(P),C_{1}(P,f),C_{2}(P,f)$ with the bounds given in Theorem \ref{th:
DriftBounds} and Proposition \ref{pr: Compl}.
The ``actual values'' of $C_{0}(P),C_{1}(P,f),C_{2}(P,f)$ are computed
via a
long Monte Carlo simulation (in which we identified regeneration epochs).
The bound for $C_{1}(P,f)$ in Theorem \ref{th: DriftBounds}(iii) depends
on $\xi V$, which is typically known, because
usually simulation starts from a deterministic initial point, say $x_0$
(in our experiments, we put $x_0=0$).
As for $C_{2}(P,f)$, its actual value varies with $n$. However, in our
experiments the dependence on $n$ was negligible and has been ignored
(the differences were within the accuracy of the reported computations,
provided that $n\geq10$).

\begin{figure}

\includegraphics{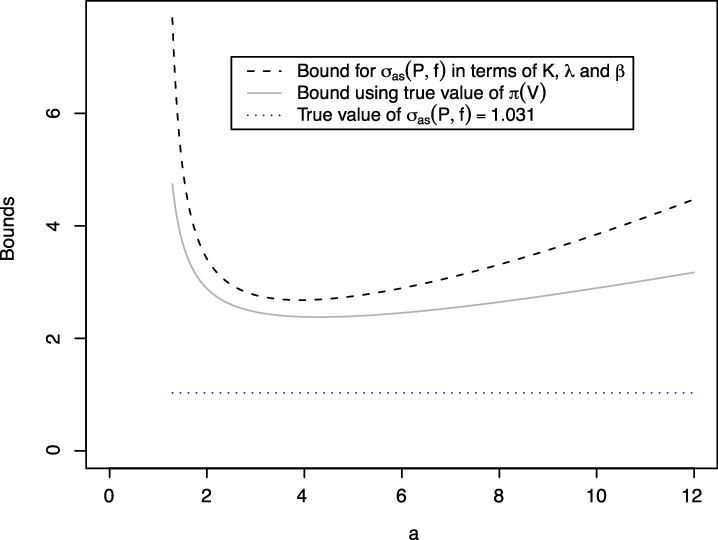}

\caption{Bounds for the root asymptotic variance $\sigma_{\mathrm{as}}(P,f)$
as functions of $a$.}\label{fig2}
\end{figure}

\begin{table}
\caption{Values of $\sigma_{\mathrm{as}}(P,f)$ vs. bounds of Theorem
\protect\ref{th:
DriftBounds}(ii) combined with
Proposition \protect\ref{pr: Compl}(i) and (ii) for different values of $t$}\label{tab1}
\begin{tabular*}{\textwidth}{@{\extracolsep{\fill}}llll@{}}
\hline
$t$& $\sigma_{\mathrm{as}}(P,f)$ & Bound with known $\pi V$ & Bound involving only $\lambda$, $K$, $\beta$ \\
\hline
\phantom{00}5 & 1.581 & 6.40 & 11.89 \\
\phantom{0}50 & 1.031 & 2.38 & \phantom{0}2.68 \\
500& 1.003 & 2.00 & \phantom{0}2.08 \\
\hline
\end{tabular*}
\end{table}

\begin{table}
\caption{Values of the constants appearing in Theorem \protect\ref{th:BasicMSE} vs. bounds of Theorem \protect\ref{th: DriftBounds} combined with
Proposition \protect\ref{pr: Compl}}\label{tab2}
\begin{tabular*}{\textwidth}{@{\extracolsep{\fill}}llll@{}}
\hline
 Constant& Actual value& Bound with known $\pi V$ & Bound involving only $\lambda$, $K$, $\beta$ \\
\hline
$C_{0}(P)$ & 0.568 & 1.761 & 2.025 \\
$C_{1}(P,f)$ & 0.125 & -- & 2.771 \\
$C_{2}(P,f)$ & 1.083 & -- & 3.752 \\
\hline
\end{tabular*}
\end{table}

Finally, let us compare the actual values of the root mean square
error, $\mathrm{RMSE}: =  \sqrt{\mathbb{E}_{\xi} (\hat\theta_{n}-\theta)^2}$,
with the bounds
given in Theorem \ref{th:BasicMSE}.
In column (a), we use the formula \eqref{eq:BasicMSE_bound} with
``true'' values of $\sigma_{\mathrm{as}}(P,f)$ and
$C_{0}(P),C_{1}(P,f),C_{2}(P,f)$ given
by \eqref{eqdef:asvar} and \eqref{eqdef:ctwo}.
Column (b) is obtained by replacing those constants by their bounds
given in Theorem \ref{th: DriftBounds} and
using the true value of $\pi V$. Finally, the bounds involving only
$\lambda$, $K$, $\beta$ are in column (c).

\begin{table}
\caption{RMSE, its bound in Theorem \protect\ref{th:BasicMSE} and further
bounds based Theorem \protect\ref{th: DriftBounds} combined with
Proposition \protect\ref{pr: Compl}}\label{tab3}
\begin{tabular*}{\textwidth}{@{\extracolsep{\fill}}lllll@{}}
\hline
&  & \multicolumn{3}{l@{}}{Bound \eqref{eq:BasicMSE_bound} } \\[-5pt]
&  & \multicolumn{3}{l@{}}{\hrulefill} \\
$n$ & $\sqrt{n}$ RMSE& (a) & (b) & (c) \\
\hline
\phantom{000,}10 & 0.98 & 1.47 & 4.87 & 5.29\\
\phantom{000,}50 & 1.02 & 1.21 & 3.39 & 3.71\\
\phantom{00,}100 & 1.03 & 1.16 & 3.08 & 3.39\\
\phantom{0,}1000 & 1.03 & 1.07 & 2.60 & 2.89\\
\phantom{0,}5000 & 1.03 & 1.05 & 2.48 & 2.77\\
10,000 & 1.03 & 1.04 & 2.45 &2.75\\
50,000 & 1.03 & 1.04 & 2.41 & 2.71\\
\hline
\end{tabular*}
\end{table}

Table \ref{tab3} clearly shows that the inequalities in Theorem \ref
{th:BasicMSE} are quite sharp.
The bounds on RMSE in column (a) become almost exact for large $n$.
However, the bounds on the constants in terms of minorization/drift
parameters are far from being tight.
While constants $C_{0}(P),C_{1}(P,f),C_{2}(P,f)$ have relatively small
influence, the problem of bounding $\sigma_{\mathrm{as}}(P,f)$ is
of primary importance.

This clearly identifies the bottleneck of the approach:
the bounds on $\sigma_{\mathrm{as}}(P,f)$ under drift condition in
Theorem \ref{th:
DriftBounds} and Proposition \ref{pr: Compl}
can vary widely in their sharpness in specific examples. We conjecture
that this may be the case in general for any bounds derived under
drift conditions. Known bounds on the rate of convergence (e.g., in
total variation norm) obtained under drift conditions are typically very
conservative, too (e.g., \cite{baxendale2005renewal,roberts1999bounds,jones2004sufficient}).
However, at present, drift conditions remain the main and most
universal tool for proving computable bounds for Markov chains on
continuous spaces.
An alternative might be working with conductance but to the best of our
knowledge, so far this approach has been applied successfully only to examples
with compact state spaces (see, e.g., \cite{rudolf2009explicit,mathe2007simple} and references therein).

\subsection{A Poisson--Gamma model} \label{subsec:pg}

Consider a hierarchical Bayesian model applied to a well-known pump
failure data set
and analysed in several papers (e.g., \cite{gelfand1990sampling,tierney1994markov,mykland1995regeneration,rosenthal1995minorization}).
Data are available for example, in \cite{MR1998913}, R package
``SMPracticals'' or in the cited Tierney's paper. They consist of
$m=10$ pairs $(y_i,t_i)$ where
$y_i$ is the number of failures for $i$th pump, during $t_i$ observed hours.
The model assumes that:
\begin{eqnarray*}
y_{i} &\sim&\operatorname{Poiss}(t_i\phi_i),
\qquad \mbox{conditionally independent for } i=1,\ldots, m,
\\
\phi_i &\sim&\operatorname{Gamma}(\alpha,r), \qquad \mbox
{conditionally i.i.d. for } i=1,\ldots,m,
\\
r & \sim&\operatorname{Gamma}(\sigma,\gamma).
\end{eqnarray*}
The posterior distribution of parameters $\phi=(\phi_1,\ldots,\phi
_m)$ and $r$ is
\[
p(\phi,r|y) \propto \Biggl(\prod
_{i=1}^{m}
\phi_i^{y_i} \mathrm{e}^{-t_i\phi_i} \Biggr)\cdot\Biggl(\prod
\limits
_{i=1}^{m}r^{\alpha}
\phi_i^{\alpha-1}\cdot \mathrm{e}^{-r\phi_i} \Biggr) \cdot
r^{\sigma-1} \mathrm{e}^{-\gamma r},
\]
where $\alpha$, $\sigma$, $\gamma$ are known hyperparameters.
The Gibbs sampler updates cyclically $r$ and $\phi$ using the
following conditional distributions:
\begin{eqnarray*}
r | \phi,y &\sim& \operatorname{Gamma} \Bigl(m\alpha+\sigma, \gamma+\sum
\phi_i \Bigr),
\\
\phi_i|\phi_{-i},r,y &\sim& \operatorname{Gamma}
(y_i+\alpha,t_i+r ).
\end{eqnarray*}
In what follows, the numeric results correspond to the same
hyperparameter values
as in the above cited papers: $\alpha=1.802$, $\sigma=0.01$ and
$\gamma=1$.
For these values, Rosenthal in \cite{rosenthal1995minorization} \mbox{constructed}
a small set $J=\{(\phi,r)\dvt 4\leq\sum\phi_i\leq9\}$ which satisfies
the one-step minorization condition
(our Assumption~\ref{as:SmallSetCond}) and established a geometric
drift condition towards $J$ (our Assumption~\ref{as: GeoDriftCond}) with
$V(\phi,r)=1+(\sum\phi_i-6.5)^2$.
The minorization and drift constants were the following:
\[
\beta=0.14,\qquad  \lambda=0.46,\qquad  K=3.3.
\]
Suppose we are to estimate the posterior expectation of a component
$\phi_i$. To get a bound on
the (root-) MSE of the MCMC estimate, we combine Theorem~\ref
{th:BasicMSE} with Proposition~\ref{th: DriftBounds} and
Proposition~\ref{pr: Compl}.
Suppose we start simulations at a point with $\sum\phi_i=6.5$ that
is, with initial value of $V$ equal to 1.
To get a better bound on $\Vert \bar{f}\Vert _{V^{1/2}}$ via
Proposition~\ref{pr: Compl}(v), we first reduce $\Vert f\Vert _{V^{1/2}}$ by a vertical shift, namely we put $f(\phi,r)=\phi_i-b$ for
$b=3.327$ (expectation
of $\phi_i$ can be immediately recovered from that of $\phi_i-b$).
Elementary and easy calculations
show that $\Vert f\Vert _{V^{1/2}}\leq3.327$. We also use the bound
taken from Proposition~\ref{pr: Compl}(ii) for
$\pi(V)$ and the inequality $\pi(V^{1/2})\leq\pi(V)^{1/2}$. Finally, we obtain
the following values of the constants:
\[
\sigma_{\mathrm{as}}(P,f)\leq 171.6\quad  \mbox{and}\quad
C_{0}(P)\leq 27.5,\qquad  C_{1}(P,f)\leq 547.7,\qquad
C_{2}(P,f)\leq 676.1.
\]

\subsection{Contracting normals} \label{subsec:cn}

As discussed in the \hyperref[sec1]{Introduction}, the results of the present paper
improve over earlier MSE bounds of \cite{eps_alpha} for geometrically
ergodic chains in that they are much more generally applicable and also
tighter. To illustrate the improvement in tightness, we analyze the MSE
and confidence estimation for the contracting normals toy-example
considered in~\cite{eps_alpha}.

For the Markov chain transition kenel
\[
P(x, \cdot) = \mathrm{N}\bigl(c x, 1- c^2\bigr),\qquad  \mbox{with}
|c|<1,  \mbox{ on } \mathcal{X} = \mathbb{R},
\]
with stationary distribution $\mathrm{N}(0,1),$ consider estimating the mean,
that is, put $f(x) = x$. Similarly as in \cite{eps_alpha} we take a drift
function $V(x) = 1+ x^2$ resulting in $\Vert f\Vert _{V^{1/2}}=1$. With the
small set $J =[-d,d]$ with $d>1,$ the drift and regeneration parameters
can be identified as
\begin{eqnarray*}
\lambda&=& c^2 + \frac{1\bigl(1-c^2\bigr) }{1+d^2}
< 1,\qquad  K = 2+c^2\bigl(d^2-1\bigr),\\
 \beta&=& 2
\biggl[\Phi\biggl(\frac{(1+|c|)d }{\sqrt{1-c^2}}\biggr) - \Phi
\biggl(
\frac{|c|d }{\sqrt{1-c^2}} \biggr)\biggr],
\end{eqnarray*}
where $\Phi$ stands for the standard normal c.d.f. We refer to \cite
{eps_alpha,baxendale2005renewal} for details on these elementary calculations.

To compare with the results of \cite{eps_alpha}, we aim at confidence
estimation of the mean. First, we combine Theorem~\ref{th:BasicMSE}
with Proposition~\ref{th: DriftBounds} and Proposition~\ref{pr:
Compl} to upperbound the MSE of $\hat{\theta}_n$ and next we use the
Chebyshev inequality. We derive the resulting minimal simulation length
$n$ guaranteeing
\[
\mathbb{P}\bigl(|\hat{\theta}_n -\theta| < \varepsilon\bigr) > 1-\alpha,
\qquad \mbox{with } \varepsilon= \alpha= 0.1.
\]
This is equivalent to finding minimal $n$ s.t.
\[
\operatorname{MSE}(\hat{\theta}_n) \leq\varepsilon^2
\alpha.
\]
Note that for small values of $\alpha$ a median trick can be applied
resulting in an exponentially tight bounds, see \cite
{niemiro2009fixed,eps_alpha,latuszynski2011nonasymptotic} for
details. The value of $c$ is set to $0.5$ and the small set half width
$d$ has been optimised numerically for each method yielding $d=1.6226$
for the bounds from \cite{eps_alpha} and $d= 1.7875$ for the results
based on our Section~\ref{Sec:Geom}. The chain is initiated at $0$,
that is, $\xi= \delta_0.$ Since in this setting the exact
distribution of $\hat{\theta}_n$ can be computed analytically, both
bounds are compared to reality, which is the exact true simulation
effort required for the above confidence estimation.

As illustrated by Table \ref{tab4}, we obtain an improvement of $5$ orders of
magnitude compared to \cite{eps_alpha} and remain less then $2$ orders
of magnitude off the truth.

\begin{table}
\caption{Comparison of the total simulation effort $n$ required for
nonasymptotic confidence estimation $ \mathbb{P}(|\hat{\theta}_n
-\theta| < \varepsilon) > 1-\alpha$ with $\varepsilon= \alpha= 0.1$
and the target function $f(x) = x$}\label{tab4}
\begin{tabular*}{\textwidth}{@{\extracolsep{\fill}}llll@{}}
\hline
Bound involving only $\lambda$, $K$, $\beta$ & Bound with known $\pi
V$ & Bound from \cite{eps_alpha} & Reality \\
\hline
77,285 & 43,783 & 6,460,000,000 & 811 \\
\hline
\end{tabular*}
\end{table}

\section{Preliminary lemmas} \label{Sec:Lemmas}

Before we proceed to the proofs for Sections \ref{Sec:Geom} and \ref
{Sec:Poly},
we need some auxiliary results that might be of independent interest.

We work under Assumptions \ref{as:SmallSetCond} (small set) and \ref
{as:DriftPoly} (the drift condition).
Note that Assumption~\ref{as: GeoDriftCond} is the special case of Assumption \ref
{as:DriftPoly}, with $\alpha=1$.
Assumption \ref{as: GeoDriftCond} implies
%
\begin{equation}
\label{eq: Drift} PV^{1/2}(x) \leq%
\cases{ \lambda^{1/2} V^{1/2}(x),& \quad for $x\notin J$,
\cr
K^{1/2},& \quad for $x\in J$,
}
\end{equation}
because by Jensen's inequality $PV^{1/2}(x)\leq\sqrt{PV(x)}$.
Whereas for $\alpha< 1,$ Lemma 3.5 of \cite{jarner2002polynomial}
for all $\eta\leq1$ yields
%
\begin{equation}
\label{eq:drift_app_eta} PV^\eta(x)
\leq%
\cases{ V^\eta(x)-\eta(1-\lambda)
V(x)^{\eta+\alpha-1},& \quad for $x\notin J$,
\cr
K^\eta,& \quad for $x\in J$.
}
\end{equation}

The following lemma is a well-known fact which appears for example, in
\cite{nummelin2002mc} (for bounded~$g$).
The proof for nonnegative function $g$ is the same.

\begin{lem}\label{lem: SquareBlock}
If $g\geq0$, then
\[
\mathbb{E}_{\nu}\Xi(g)^{2}=\mathbb{E}_{\nu} T \Biggl(
\mathbb{E}_\pi g(X_0)^2+2\sum
_{n=1}^\infty\mathbb{E}_\pi g(X_0)g(X_n)
\mathbb{I}(T >n) \Biggr).
\]
\end{lem}

We shall also use the generalised Kac lemma, in the following form that
follows as an easy corollary from Theorem 10.0.1 of
\cite{meyn1993markov}.

\begin{lem} If $\pi(|f|) < \infty,$ then \label{lem:Kac}
%
\begin{eqnarray}\label{eq:def_tau}
\pi(f) &= & \int_J \mathbb{E}_{x} \sum
_{i=1}^{\tau(J)}f(X_i) \pi(
\mathrm{d}x),\qquad  \mbox{where}
\nonumber\\[-8pt]\\[-8pt]
\tau(J) & := & \min\{n > 0\dvt X_n \in J\}.\nonumber
\end{eqnarray}
\end{lem}

The following lemma is related to other calculations in the drift
conditions setting, for example, \cite{baxendale2005renewal,lund1996geometric,douc2004quantitative,rosenthal2002quantitative,fort2003computable,douc2008bounds}.

\begin{lem}\label{lem: Bax} If Assumptions \ref{as:SmallSetCond} and
\ref{as:DriftPoly} hold, then for all $\eta\leq1$
\begin{eqnarray*}
\mathbb{E}_x \sum_{n=1}^{T-1}V^{\alpha+\eta-1}(X_n)&
\leq& \frac{V^\eta
(x)-1+\eta(1-\lambda) -\eta(1-\lambda)V^{\alpha+\eta-1}(x)}{\eta
(1-\lambda)}\mathbb{I}(x\notin J)
\\
&&{} + \frac{K^\eta-1}{\beta\eta(1-\lambda)}+ \frac{1 }{\beta}-1
\\
&\leq& \frac{V^\eta(x) }{\eta(1-\lambda)}+ \frac{K^\eta-1- \beta
}{\beta\eta(1-\lambda)}+\frac{1 }{\beta}-1 \qquad \mbox{(if
additionally }\alpha+ \eta\geq1).
\end{eqnarray*}
\end{lem}

\begin{cor} \label{cor_after_main_lemma} For $ \mathbb{E}_x \sum
_{n=0}^{T-1}V^{\alpha+\eta-1}(X_n)$, we need to add the term
$V^{\alpha+\eta-1}(x)$. Hence,
%
\begin{eqnarray*}
\mathbb{E}_x \sum_{n=0}^{T-1}V^{\alpha+\eta-1}(X_n)&
\leq& \frac{V^\eta
(x)-1+\eta(1-\lambda) -\eta(1-\lambda)V^{\alpha+\eta-1}(x)}{\eta
(1-\lambda)}
\\
&&{} 
+ \frac{K^\eta-1}{\beta\eta(1-\lambda)}+ {1 \over\beta
}-1+V^{\alpha+\eta-1}(x)
\\
&=& \frac{V^\eta(x) }{\eta(1-\lambda)}+ \frac{K^\eta-1-\beta
}{\beta\eta(1-\lambda)}+{1 \over\beta}.
\end{eqnarray*}
\end{cor}

In the case of geometric drift, the second inequality in Lemma~\ref
{lem: Bax} can be replaced by a slightly better bound.
For $\alpha=\eta=1$, the first inequality in Lemma~\ref{lem: Bax}
entails the following.

\begin{cor}\label{geoBax} If Assumptions \ref{as:SmallSetCond} and
\ref{as: GeoDriftCond} hold, then
\[
\mathbb{E}_x \sum_{n=1}^{T-1}V(X_n)
\leq\frac{\lambda
V(x)}{1-\lambda}+ \frac{K-\lambda-\beta}{\beta(1-\lambda)}.
\]
\end{cor}

\begin{pf*}{Proof of Lemma \ref{lem: Bax}}
The proof is given for $\eta=1,$ because for $\eta< 1$ it is
identical and the constants can be obtained from \eqref{eq:drift_app_eta}.

Let $S:=S_{0}:=\min\{n\geq0\dvt X_{n}\in J\}$ and $S_{j}:=\min\{
n>S_{j-1}\dvt X_{n}\in J\}$
for $j=1,2,\ldots.$ Moreover, set
\begin{eqnarray*}
  H(x) &:=& \mathbb{E}_{x}\sum
_{n=0}^{S}V^\alpha(X_{n}),
\\
\tilde{H} &:=& \sup_{x\in J} \mathbb{E}_{x} \Biggl(\sum
_{n=1}^{S_{1}} V^\alpha(X_{n})
\Big|\Gamma_{0}=0 \Biggr) = \sup_{x\in J} \int Q(x,
\mathrm{d}y)H(y).
\end{eqnarray*}
Note that $H(x)=V^\alpha(x)$ for $x\in J$ and recall that $Q$ denotes
the normalized ``residual kernel'' defined in Section \ref{Sec:Regen}.

We will first show that
%
\begin{equation}
\label{eq:BlazejLem} H(x) \leq\frac{V(x)-\lambda}{1-\lambda}
\qquad \mbox{for } x\in\mathcal{X}.
\end{equation}

Let $\mathcal{F}_n=\sigma(X_0,\ldots,X_n)$ and remembering that
$\eta= 1,$
rewrite \eqref{eq:drift_app_eta} as
%
\begin{equation}
\label{as:drift_app1}
V(X_n)^\alpha
\mathbb{I}(X_n\notin J) \leq\frac{1}{1-\lambda
}
\bigl[V(X_n)-\mathbb{E}\bigl(V(X_{n+1})|\mathcal{F}_n\bigr)
\bigr]\mathbb{I}(X_n\notin J).
\end{equation}
Fix $x\notin J$. Since $\{X_n\notin J\}\supseteq\{S>n\}\in\mathcal{F}_n$,
we can apply \eqref{as:drift_app1} and write
\begin{eqnarray*}
  \mathbb{E}_x \sum
_{n=0}^{(S-1)\land m}V^\alpha(X_{n}) &=&
\mathbb{E}_x\sum_{n=0}^{m}V^\alpha(X_{n})
\mathbb{I}(S>n)
\\
& \leq&\frac{1}{1-\lambda}\sum_{n=0}^{m}
\mathbb{E}_x \bigl[V(X_n)-\mathbb{E}\bigl(V(X_{n+1})|
\mathcal{F}_n\bigr) \bigr]\mathbb{I}(S>n)
\\
& = &\frac{1}{1-\lambda}\sum_{n=0}^{m}
\bigl[\mathbb{E}_x V(X_n)\mathbb{I}(S>n)-\mathbb{E}_x\mathbb{E}
\bigl(V(X_{n+1})\mathbb{I}(S>n)|\mathcal{F}_n\bigr) \bigr]
\\
& =& \frac{1}{1-\lambda}\sum_{n=0}^{m}
\bigl[ \mathbb{E}_x V(X_n)\mathbb{I}(S>n)-\mathbb{E}_x
V(X_{n+1})\mathbb{I}(S>n+1)
\\
&&\hspace*{42pt}{} -
\mathbb{E}_x V(X_{n+1})\mathbb{I}(S=n+1) \bigr]
\\
& \leq&\frac{1}{1-\lambda} \Biggl[V(x)-\mathbb{E}_x
V(X_{m+1})\mathbb{I}(S>m+1)\\
&&\hspace*{31pt}{}-\sum_{n=0}^{m}
\mathbb{E}_x V(X_{n+1})\mathbb{I}(S=n+1) \Biggr]
\\
& = &\frac{V(x)-\mathbb{E}_x V(X_{S\land(m+1)})}{1-\lambda},
\end{eqnarray*}
so
\begin{eqnarray*}
  \mathbb{E}_x \sum
_{n=0}^{S\land(m+1)}V^\alpha(X_{n}) &=&
\mathbb{E}_x \sum_{n=0}^{(S-1)\land m}V^\alpha(X_{n})+
\mathbb{E}_x V^{\alpha}(X_{S\land(m+1)})
\\
&\leq&\frac{V(x)-\mathbb{E}_x V(X_{S\land(m+1)})}{1-\lambda
}+\mathbb{E}_x V(X_{S\land(m+1)})
\\
&=& \frac{V(x)-\lambda\mathbb{E}_x V(X_{S\land
(m+1)})}{1-\lambda} \leq\frac{V(x)-\lambda}{1-\lambda}.
\end{eqnarray*}
Letting $m \to\infty$ yields equation \eqref{eq:BlazejLem} for
$x\notin J$. For $x\in J$, \eqref{eq:BlazejLem} is obvious.

Next, from Assumption \ref{as:DriftPoly} we obtain $PV(x)=(1-\beta
)QV(x)+\beta\nu V\leq K$ for $x\in J$, so
$QV(x)\leq(K-\beta)/(1-\beta)$ and, taking into account \eqref
{eq:BlazejLem},
%
\begin{equation}
\label{eq: HJ} \tilde{H} \leq\frac{(K-\beta)/(1-\beta)-\lambda
}{1-\lambda} =
\frac{K-\lambda-\beta(1-\lambda)}{(1-\lambda)(1-\beta)}.
\end{equation}
Recall that $T:=\min\{n\geq1\dvt \Gamma_{n-1}=1\}$. For $x\in J$, we
thus have 
\begin{eqnarray*}
\mathbb{E}_x\sum_{n=1}^{T-1}V^\alpha(X_n)
&=&\mathbb{E}_x \sum_{j=1}^{\infty}\sum
_{n=S_{j-1}+1}^{S_{j}} V^\alpha(X_{n})
\mathbb{I}(\Gamma_{S_{0}}=\cdots=\Gamma_{S_{j-1}}=0)
\\
&=& 
\sum_{j=1}^{\infty}
\mathbb{E}_x \Biggl(\sum_{n=S_{j-1}+1}^{S_{j}}
V^\alpha(X_{n}) \Big\vert\Gamma_{S_{0}}=\cdots=
\Gamma_{S_{j-1}}=0 \Biggr) (1-\beta)^{j}
\\
&\leq&\sum_{j=1}^{\infty} \tilde{H}(1-
\beta)^{j} \leq\frac{K-\lambda}{\beta(1-\lambda)}-1,
\end{eqnarray*}
by \eqref{eq: HJ}. For $x\notin J$, we have to add one more term and
note that the above calculation also applies.
\[
\mathbb{E}_x\sum
_{n=1}^{T-1}V^\alpha(X_n) =
 \mathbb{E}_x \sum_{n=1}^{S_{0}}
V^\alpha(X_{n}) 
+\mathbb{E}_x \sum
_{j=1}^{\infty} \sum_{n=S_{j-1}+1}^{S_{j}}
V^\alpha(X_{n}) \mathbb{I}(\Gamma_{S_{0}}=\cdots=
\Gamma_{S_{j-1}}=0).
\]
The extra term is equal to $H(x)-V^\alpha(x)$ and we use \eqref
{eq:BlazejLem} to bound it.
Finally, we obtain
%
\begin{equation}
\label{eq: final} \mathbb{E}_x \sum_{n=1}^{T-1}V^{\alpha}(X_n)
\leq\frac
{V(x)-\lambda-(1-\lambda)V^\alpha(x)}{1-\lambda}\mathbb{I}(x\notin J)+
\frac{K-\lambda}{\beta(1-\lambda)}-1.
\end{equation}
\end{pf*}

\begin{lem}\label{lem:Poly}
If Assumptions \ref{as:SmallSetCond} and
\ref{as:DriftPoly} hold, then
\begin{enumerate}[(iii)]
\item[(i)] for all $\eta\leq\alpha$
\[
\pi\bigl(V^\eta\bigr) \leq\biggl(\frac{K-\lambda}{1-\lambda
}
\biggr)^{{\eta/\alpha}},
\]
\item[(ii)]
\[
\pi(J) \geq\frac{1-\lambda}{K-\lambda},
\]
\item[(iii)] for all $n\geq0$ and $\eta\leq\alpha$
\[
\mathbb{E}_\nu V^\eta(X_n) \leq
\frac{1}{\beta^{{\eta
}/{\alpha}}} \biggl(\frac{K-\lambda}{1-\lambda} \biggr)^{2{\eta
}/{\alpha}}.
\]
\end{enumerate}
\end{lem}

\begin{pf}
It is enough to prove (i) and (iii) for $\eta=\alpha$ and apply the
Jensen inequality for $\eta<\alpha$.
We shall need an upper bound on $E_x\sum_{n=1}^{\tau(J)} V^\alpha
(X_n)$ for $x\in J$, where $\tau(J)$ is defined in \eqref
{eq:def_tau}. From the proof of Lemma \ref{lem: Bax},
\[
\mathbb{E}_x\sum_{n=1}^{\tau(J)}
V^\alpha(X_n) = PH(x) \leq
\frac{K-\lambda}{1-\lambda},\qquad  x\in J.
\]
And by Lemma \ref{lem:Kac}, we obtain
\[
1 \leq\pi V^\alpha= \int_{J}
\mathbb{E}_x\sum_{n=1}^{\tau(J)}V^\alpha(X_n)
\pi(\mathrm{d}x) \leq\pi(J)\frac
{K-\lambda}{1-\lambda},
\]
which implies (i) and (ii).

By integrating the small set Assumption \ref{as:SmallSetCond} with
respect to $\pi$ and from (ii) of the current lemma, we obtain
\[
\frac{\mathrm{d}\nu}{\mathrm{d}\pi} \leq\frac{1}{\beta\pi(J)} \leq
\frac{K-\lambda}{\beta(1-\lambda)}.
\]
Consequently,
\begin{eqnarray*}
\mathbb{E}_\nu V^\alpha(X_n) & =& \int
_{\mathcal{X}
}P^nV^\alpha(x)\frac{\mathrm{d}\nu}{\mathrm{d}\pi}
\pi(\mathrm{d}x) \leq\frac
{K-\lambda}{\beta(1-\lambda)}\int_{\mathcal{X}}P^nV^\alpha(x)
\pi(\mathrm{d}x )
\\
& = & \frac{K-\lambda}{\beta(1-\lambda)}\pi\bigl(V^\alpha\bigr),
\end{eqnarray*}
and (iii) results from (i).
\end{pf}

\section{\texorpdfstring{Proofs for Section \protect\ref{Sec:Geom} and \protect\ref{Sec:Poly}}{Proofs for Section 4 and 5}}
\label{Sec:Proofs}

In the proofs for Section \ref{Sec:Geom}, we work under Assumption
\ref{as: GeoDriftCond} and repeatedly use
Corollary~\ref{geoBax}.

\begin{pf*}{Proof of Theorem \ref{th: DriftBounds}}
\begin{longlist}[(iii)]
\item[(i)] Recall that $C_{0}(P)=\mathbb{E}_{\pi} T-\frac{1}{2}$, write
\[
\mathbb{E}_{\pi} T \leq1+ \mathbb{E}_{\pi} \sum
_{n=1}^{T-1}V(X_n)
\]
and use Corollary \ref{geoBax}. The proof of the alternative statement
(i$\prime$) uses first \eqref{eq: Drift} and then is the same.

\item[(ii)] Without loss of generality, assume that $ \Vert{\bar f}\Vert
_{V^{1/2}}= 1$.
By Lemma \ref{lem: SquareBlock}, we then have
\begin{eqnarray*}
  \sigma_{\mathrm{as}}^{2}(P,f) &
=& \mathbb{E}_\nu\bigl(\Xi({\bar f})\bigr)^2/
\mathbb{E}_{\nu} T\leq\mathbb{E}_{\nu} \bigl(\Xi\bigl(V^{1/2}
\bigr)\bigr)^2/\mathbb{E}_{\nu} T
\\
 & =& \mathbb{E}_\pi V(X_0)+2\mathbb{E}_\pi\sum
_{n=1}^{T-1}V^{1/
2}(X_0)V^{1/2}(X_n)
=: \mathrm{I} + \mathrm{II}.
\end{eqnarray*}
To bound the second term, we will use Corollary \ref{geoBax} with
$V^{1/2}$ in place of $V$, which is legitimate because of \eqref
{eq: Drift}.
\begin{eqnarray*}
\mathrm{II}/2 &=&
\mathbb{E}_\pi\sum_{n=1}^{T-1}V^{1/2}(X_0)V^{1/2}(X_n)
= \mathbb{E}_\pi V^{1/2}(X_0)\mathbb{E}\Biggl(\sum
_{n=1}^{T-1}V^{1/2}(X_n)\Big|X_0
\Biggr)
\\
& \leq&\mathbb{E}_{\pi} V^{1/2}(X_0) \biggl(
\frac{\lambda^{1/2}}{1-\lambda^{1/2}}V^{1/2}(X_0)+\frac{K^{1/2}-\lambda^{1/2}-\beta}{\beta
(1-\lambda^{1/2})} \biggr)
\\
&=& \frac{\lambda^{1/2}}{1-\lambda^{1/2}}\pi(V)+\frac
{K^{1/2}-\lambda^{1/2}-\beta}{\beta(1-\lambda^{1/2})}\pi\bigl(V^{1/2}
\bigr).
\end{eqnarray*}
Rearranging terms in $\mathrm{I}+\mathrm{II}$, we obtain
\[
\sigma_{\mathrm{as}}^{2}(P,f) \leq\frac{1
+\lambda^{1/2}}{1-\lambda^{1/2}}
\pi(V)+\frac{2(K^{1/2}-\lambda^{1/2}-\beta)}{\beta(1-\lambda^{1/2})}\pi\bigl
(V^{1/2}\bigr)
\]
and the proof of (ii) is complete.

\item[(iii)]
The proof is similar to that of (ii) but more delicate, because
we now cannot use Lemma~\ref{lem: SquareBlock}. First, write
\begin{eqnarray*}
\mathbb{E}_x \bigl(\Xi
\bigl(V^{1/2}\bigr)\bigr)^2 &=& \mathbb{E}_x
\Biggl(\sum_{n=0}^{T-1}V^{1/2}(X_n)
\Biggr)^2 = \mathbb{E}_x \Biggl(\sum
_{n=0}^\infty V^{1/2}(X_n)\mathbb{I}
(n<T) \Biggr)^2
\\
&=& \mathbb{E}_x\sum_{n=0}^\infty
V(X_n)\mathbb{I}(n<T) +2\mathbb{E}_x \sum
_{n=0}^\infty\sum_{j=n+1}^\infty
V^{1/2}(X_n)V^{1/2}(X_j)\mathbb{I}(j<T )
\\
&=:& \mathrm{I} + \mathrm{II}.
\end{eqnarray*}
The first term can be bounded directly using Corollary \ref{geoBax}
applied to $V$.
\[
\mathrm{I} = \mathbb{E}_x\sum
_{n=0}^\infty V(X_n)\mathbb{I}(n<T) \leq
\frac{1}{1-\lambda} V(x)+\frac{K-\lambda-\beta}{\beta(1-\lambda)}.
\]

To bound the second term, first condition on $X_n$ and apply Corollary
\ref{geoBax} to $V^{1/2}$,
then again apply this corollary to $V$ and to $V^{1/2}$.
\begin{eqnarray*}
\mathrm{II}/2& =&
\mathbb{E}_x \sum_{n=0}^\infty
V^{1/2}(X_n)\mathbb{I}(n<T ) \mathbb{E}\Biggl(\sum
_{j=n+1}^\infty V^{1/2}(X_j)
\mathbb{I}(j<T ) \Big|X_n \Biggr)
\\
& \leq&\mathbb{E}_x \sum_{n=0}^\infty
V^{1/2}(X_n)\mathbb{I}(n<T ) \biggl( \frac{\lambda^{1/2}}{1-\lambda^{1/2}}V^{1/2}(X_n)
+\frac{K^{1/2}-\lambda^{1/2}-\beta}{\beta(1-\lambda
^{1/2})} \biggr)
\\
& =& \frac{\lambda^{1/2}}{1-\lambda^{1/2}}\mathbb{E}_x
\sum
_{n=0}^\infty V(X_n)\mathbb{I}(n<T ) 
+ \frac{K^{1/2}-\lambda^{1/2}-\beta}{\beta(1-\lambda
^{1/2})}\mathbb{E}_x \sum_{n=0}^\infty
V^{1/2}(X_n)\mathbb{I}(n<T )
\\
& \leq&\frac{\lambda^{1/2}}{1-\lambda^{1/2}} \biggl(
\frac{1}{1-\lambda}V(x) +
\frac{K-\lambda-\beta}{\beta(1-\lambda)} \biggr)
\\
&&{} %
+ \frac{K^{1/2}-\lambda^{1/2}-\beta}{\beta(1-\lambda
^{1/2})} \biggl(
\frac{1}{1-\lambda^{1/2}}V^{1/2}(x) +\frac{K^{1/2}-\lambda^{1/2}-\beta}{\beta(1-\lambda^{1/2})} \biggr).
\end{eqnarray*}
Finally, rearranging terms in $\mathrm{I}+\mathrm{II}$, we obtain
\begin{eqnarray*}
\mathbb{E}_x \bigl(\Xi
\bigl(V^{1 /2}\bigr)\bigr)^2 & \leq&\frac{1}{(1-\lambda
^{1/2})^{2}}
V(x) + \frac{2(K^{1/2}-\lambda^{1/2}-\beta)}{\beta
(1-\lambda^{1/2})^{2}} V^{1/2}(x)
\\
&&{} + \frac{\beta(K-\lambda-\beta)+2(K^{1/2}-\lambda^{1/2}-\beta)^{2}} {
\beta^{2}(1-\lambda^{1/2})^{2}},
\end{eqnarray*}
which is tantamount to the desired result.

\item[(iv)] The proof of (iii) applies the same way.
\end{longlist}
\upqed\end{pf*}

\begin{pf*}{Proof of Proposition \ref{pr: Compl}}
For (i) and (ii) Assumption \ref{as: GeoDriftCond} or respectively
drift condition \eqref{eq: Drift} implies that $\pi V=\pi PV\leq
\lambda(\pi V-\pi(J)) +K \pi(J)$
and the result follows immediately.


(iii) and (iv) by induction: $\xi P^{n+1} V= \xi P^n (PV)\leq\xi P^n
(\lambda V+K)\leq\lambda K/(1-\lambda)+K
=K/(1-\lambda)$.


(v) We compute:
\begin{eqnarray*}
\Vert{\bar f}\Vert_{V}& = & \sup_{x \in\mathcal{X}} \frac{|f(x) -
\pi
f|}{V(x)}
\leq\sup_{x \in\mathcal{X}} \frac{|f(x)| +
|\pi f|}{V(x)} \leq
\Vert f \Vert_{V}+ \sup_{x \in\mathcal{X}}
\frac{ \pi( {(|f|/
V)}V)}{V(x)}
\\
& \leq& \sup_{x \in\mathcal{X}} \biggl(\Vert f \Vert_{V} \biggl
[1 +
\frac{\pi V}{V(x)} \biggr] \biggr) \leq\Vert f \Vert_{V}
\biggl[1 + \frac{\pi(J)(K-\lambda
)}{(1-\lambda) \inf_{x \in\mathcal{X}}V(x)} \biggr].
\end{eqnarray*}
\upqed\end{pf*}

In the proofs for Section \ref{Sec:Poly}, we work under Assumption
\ref{as:DriftPoly} and repeatedly use
Lemma \ref{lem: Bax} or Corollary \ref{cor_after_main_lemma}.

\begin{pf*}{Proof of Theorem \ref{th: PolyDriftBounds}}
\begin{longlist}
\item[(i)] Recall that $C_{0}(P)=\mathbb{E}_{\pi} T-\frac{1}{2}$ and write
\[
\mathbb{E}_\pi T \leq1+ \mathbb{E}_\pi\sum
_{i=1}^{T-1}V^{2\alpha
-1}(X_n) =
1+ \int_{\mathcal{X}}\mathbb{E}_x\sum
_{i=1}^{T-1}V^{2\alpha-1}(X_n)\pi(\mathrm{d}x).
\]
From Lemma \ref{lem: Bax} with $V$, $\alpha$ and $\eta=\alpha$, we have
\begin{eqnarray*}
C_{0}(P)&\leq& - {1 \over2} + 1+ \int_{\mathcal{X}}
\biggl(\frac
{V^\alpha(x)-1}{\alpha(1-\lambda)}+ \frac{K^\alpha-1}{\beta\alpha
(1-\lambda)}+\frac{1}{\beta}-1 \biggr) \pi(\mathrm{d}x)
\\
&=&\frac{1}{\alpha(1-\lambda)} \pi\bigl(V^\alpha\bigr) +\frac
{K^\alpha-1 - \beta}{\beta\alpha(1-\lambda)}+{1
\over\beta}-\frac{1}{2}.
\end{eqnarray*}
\item[(ii)] Without loss of generality, we can assume that $ \Vert{\bar
f}\Vert_{V^{({3}/{2})\alpha-1}} = 1$.
By Lemma \ref{lem: SquareBlock}, we have
\begin{eqnarray*}
  \sigma_{\mathrm{as}}^{2}(P,f) &=&
\mathbb{E}_\nu\bigl(\Xi({\bar f})\bigr)^2/
\mathbb{E}_{\nu} T \leq\mathbb{E}_{\nu} \bigl(\Xi
\bigl(V^{({3}/{2})\alpha-1}\bigr)\bigr)^2/\mathbb{E}_{\nu} T
\\
&=& \mathbb{E}_\pi V(X_0)^{3\alpha-2}+2
\mathbb{E}_\pi\sum_{n=1}^{T-1}V^{({3}/{2})\alpha-1}(X_0)V^{({3}/{2})\alpha
-1}(X_n)
=: \mathrm{I} + \mathrm{II}.
\end{eqnarray*}
To bound the second term, we will use Lemma \ref{lem: Bax} with $V$,
$\alpha$ and $\eta={\alpha\over2}$.
\begin{eqnarray*}
\mathrm{II}/2 &=&
\mathbb{E}_\pi\sum_{n=1}^{T-1}V^{({3}/{2})\alpha-1}(X_0)V^{({3}/{2})\alpha-1}(X_n)
= \mathbb{E}_\pi V^{({3}/{2})\alpha-1}(X_0)\mathbb{E}
\Biggl(\sum_{n=1}^{T-1}V^{({3}/{2})\alpha-1}(X_n)\Big|X_0
\Biggr)
\\
&\leq&\mathbb{E}_{\pi} V^{({3}/{2})\alpha-1}(X_0) \biggl(
\frac{V^{{\alpha/2}}(X_0)-1}{{\alpha
}/{2}(1-\lambda)}+ \frac{K^{{\alpha/2}}-1}{\beta{{\alpha}/{2}}(1-\lambda
)}+\frac{1}{\beta}-1 \biggr)
\\
&=& \frac{2}{\alpha(1-\lambda)}\pi\bigl(V^{2\alpha-1}\bigr)+
\biggl(
\frac
{2K^{{\alpha}/{2}}-2 - 2\beta}{\alpha\beta(1-\lambda
)}+\frac{1}{\beta}-1 \biggr)\pi\bigl(V^{({3}/{2})\alpha-1}\bigr).
\end{eqnarray*}
The proof of (ii) is complete.
\item[(iii)] The proof is similar to that of (ii) but more delicate, because
we now cannot use Lemma~\ref{lem: SquareBlock}. Write
\begin{eqnarray*}
\mathbb{E}_x \bigl(\Xi
\bigl(V^{({3}/{2})\alpha-1}\bigr)\bigr)^2 &=& \mathbb{E}_x
\Biggl(\sum_{n=0}^{T-1}V^{({3}/{2})\alpha-1}(X_n)
\Biggr)^2 = \mathbb{E}_x \Biggl(\sum
_{n=0}^\infty V^{({3}/{2})\alpha-1}(X_n)\mathbb{I}
(n<T) \Biggr)^2
\\
&= &\mathbb{E}_x\sum_{n=0}^\infty
V^{3\alpha-2}(X_n)\mathbb{I}(n<T)
\\
&&{} +2\mathbb{E}_x \sum_{n=0}^\infty
\sum_{j=n+1}^\infty V^{({3}/{2})\alpha-1}(X_n)V^{({3}/{2})\alpha-1}(X_j)
\mathbb{I}(j<T )
\\
&=:& \mathrm{I} + \mathrm{II}.
\end{eqnarray*}
The first term can be bounded directly using Corollary \ref
{cor_after_main_lemma} with $\eta=2\alpha-1$
\[
\mathrm{I}= \mathbb{E}_x
\sum_{n=0}^\infty V^{3\alpha
-2}(X_n)
\mathbb{I}(n<T) 
\leq\frac{ V^{2\alpha-1}(x)}{(2\alpha-1)(1-\lambda)}+\frac
{K^{2\alpha
-1}-1 - \beta}{(2\alpha-1)\beta(1-\lambda)}+\frac{1
}{\beta}.
\]

To bound the second term, first condition on $X_n$ and use Corollary
\ref{cor_after_main_lemma} with $\eta=\frac{\alpha}{2}$
then again use Corollary \ref{cor_after_main_lemma} with $\eta=\alpha
$ and $\eta=\frac{\alpha}{2}$.
\begin{eqnarray*}
\mathrm{II}/2 &=&
\mathbb{E}_x \sum_{n=0}^\infty
V^{({3}/{2})\alpha-1}(X_n)\mathbb{I}(n<T ) \mathbb{E}\Biggl(\sum
_{j=n+1}^\infty V^{({3}/{2})\alpha-1}(X_j)
\mathbb{I}(j<T ) \Big|X_n \Biggr)
\\
&\leq&\mathbb{E}_x \sum_{n=0}^\infty
V^{({3}/{2})\alpha-1}(X_n)\mathbb{I}(n<T ) \biggl(\frac{2
V^{{\alpha}/{2}}(X_n)}{\alpha(1-\lambda)} +
\frac{2K^{{\alpha}/{2}}-2 -2\beta}{\alpha\beta(1-\lambda
)}+\frac{1}{\beta}-1 \biggr)
\\
&=& \frac{2}{\alpha(1-\lambda)} \mathbb{E}_x \sum
_{n=0}^\infty V(X_n)^{2\alpha-1}\mathbb{I}(n<T
)
\\
&&{} %
+ \biggl(\frac{2K^{{\alpha}/{2}}-2 - 2 \beta}{\alpha\beta
(1-\lambda)}+\frac{1}{\beta}-1
\biggr) \mathbb{E}_x \sum_{n=0}^\infty
V^{({3}/{2})\alpha-1}(X_n)\mathbb{I}(n<T )
\\
&\leq&\frac{2}{\alpha(1-\lambda)} \biggl(\frac{1}{\alpha
(1-\lambda)} V^{\alpha}(x)+
\frac{K^{\alpha}-1 - \beta}{\alpha
\beta(1-\lambda)} +\frac{1}{\beta} \biggr)
\\
&&{} + \biggl(\frac{2K^{{\alpha}/{2}}-2 - 2\beta}{\alpha
\beta(1-\lambda)}+\frac{1}{\beta}-1 \biggr) \biggl(
\frac{2V^{{\alpha}/{2}}(x)}{\alpha(1-\lambda)} +\frac
{2K^{{\alpha}/{2}}-2 - 2\beta}{\alpha\beta(1-\lambda
)}+\frac{1}{\beta} \biggr) .
\end{eqnarray*}
So after gathering the terms
%
\begin{eqnarray}\label{eq:proofC1Poly}
&&\mathbb{E}_x \bigl(\Xi\bigl(V^{({3}/{2})\alpha-1}\bigr)\bigr
)^2\nonumber\\
 &&\quad
\leq \frac{1}{(2\alpha-1)(1-\lambda)} V^{2\alpha-1}(x)+ \frac
{4}{\alpha^2(1-\lambda)^2}V^{\alpha}(x)
+\frac{\alpha
(1-\lambda)+4}{\alpha\beta(1-\lambda)}
\nonumber\\[-8pt]\\[-8pt]
&& \qquad {} + \biggl(\frac{8K^{\alpha/2}-8 - 8 \beta}{\alpha^2\beta
(1-\lambda)^2}+\frac{4-4\beta}{\alpha\beta(1-\lambda
)} \biggr)V^{\alpha/2}(x)
+ \frac{K^{2\alpha-1}-1 - \beta}{(2\alpha-1)\beta(1-\lambda)}
\nonumber\\
&&\qquad {} + \frac{4(K^{\alpha}-1 -\beta)}{\alpha^2\beta(1-\lambda
)^2} +2 \biggl(\frac{2K^{{\alpha}/{2}}-2 - 2\beta}{\alpha
\beta
(1-\lambda)}+\frac{1}{
\beta} \biggr)^2 -2 \biggl(\frac{2K^{{\alpha}/{2}}-2 - 2 \beta}{\alpha\beta(1-\lambda)}+\frac{1}{\beta}
\biggr).\nonumber
\end{eqnarray}
\item[(iv)] Recall that $C_{2}(P,f)^2=\mathbb{E}_{\xi} (\sum
_{i=n}^{T_{R(n)}-1}|{\bar f}(X_i)|\mathbb{I}(T< n) )^2$ and we have
%
\begin{eqnarray}
\label{eq:proofC2Poly} %
&&
\mathbb{E}_{\xi} \Biggl(\sum_{i=n}^{T_{R(n)}-1}\big|{
\bar f}(X_i)\big|\mathbb{I}(T< n) \Biggr)^2\nonumber
\\
& &\quad %
= \sum_{j=1}^{n}
\mathbb{E}_{\xi} \Biggl( \Biggl(\sum_{i=n}^{T_{R(n)}-1}\big|{
\bar f}(X_i)\big|\mathbb{I}(T< n) \Biggr)^2 \Big|T=j \Biggr)
\mathbb{P}_{\xi}(T=j)
\nonumber\\[-8pt]\\[-8pt]
&&\quad  %
\leq\sum_{j=1}^{n}
\mathbb{E}_{\nu} \Biggl( \sum_{i=n-j}^{T_{R(n-j)}-1}\big|{
\bar f}(X_i)\big| \Biggr)^{2} \mathbb{P}_{\xi}(T=j)\nonumber
\\
& &\quad %
= \sum_{j=1}^{n}
\mathbb{E}_{\nu P^{n-j}} \Biggl( \sum_{i=0}^{T-1}\big|{
\bar f} (X_i)\big| \Biggr)^{2} \mathbb{P}_{\xi}(T=j).\nonumber
\end{eqnarray}
Since
\[
\mathbb{E}_{\nu P^{n-j}} \Biggl( \sum_{i=0}^{T-1}\big|{
\bar f}(X_i)\big| \Biggr)^{2}=\nu P^{n-j} \Biggl(
\mathbb{E}_{x} \Biggl( \sum_{i=0}^{T-1}\big|{
\bar f}(X_i)\big| \Biggr)^{2} \Biggr)
\]
and $|{\bar f}|\leq V^{({3}/{2})\alpha-1}$
we put \eqref{eq:proofC1Poly} into \eqref{eq:proofC2Poly} and apply
Lemma~\ref{lem:Poly} to complete the proof.\mbox{\hfill}\qed
\end{longlist}

\noqed\end{pf*}

\begin{pf*}{Proof of Proposition \ref{pr: PolyCompl}}
For (i) see Lemma~\ref{lem:Poly}. For (ii), we compute:
\begin{eqnarray*}
\Vert{\bar f}\Vert_{V^\eta} & = & \sup_{x \in\mathcal{X}}
\frac
{|f(x) - \pi f|}{V^\eta(x)} \leq\sup_{x \in\mathcal{X}} \frac
{|f(x)| + |\pi f|}{V^\eta(x)}
\leq\Vert f \Vert_{V^\eta} + \sup_{x \in\mathcal{X}} \frac{ \pi
( {(|f|/ V^{\eta})}V^{\eta
})}{V^\eta(x)}
\\
& \leq& \sup_{x \in\mathcal{X}} \biggl(\Vert f \Vert_{V^\eta
} \biggl[1 +
\frac{\pi V^\eta}{V^\eta(x)} \biggr] \biggr) \leq\Vert f \Vert
_{V^\eta}
\bigl( 1+ \pi\bigl(V^\eta\bigr) \bigr).
\end{eqnarray*}
\upqed\end{pf*}

%

\section*{Acknowledgements}
We thank Gersende Fort and Jacek Weso{\l}owski for insightful comments
on earlier versions of this work. We also gratefully acknowledge the
help of Agnieszka Perduta and comments of three referees and an
Associate Editor that helped improve the paper.

Work partially supported by Polish Ministry of Science and Higher Education
Grants No. N~N201387234 and N N201 608740. Krzysztof \L atuszy\'
{n}ski was also partially supported by EPSRC.


%

\printhistory

\end{document}